\documentclass[11pt]{article}
\usepackage[letterpaper, margin=1in]{geometry}

\usepackage{amsmath, amssymb, amsfonts, amsthm}
\usepackage[table]{xcolor}

\usepackage{bm,bbm}
\usepackage{booktabs}
\usepackage[ruled, linesnumbered]{algorithm2e}
\usepackage{mathrsfs}
\usepackage{setspace}
\usepackage[utf8x]{inputenc}
\usepackage{algpseudocode,algorithmicx}
\algnewcommand\algorithmicinput{\textbf{Input:}}
\algnewcommand\algorithmicoutput{\textbf{Output:}}
\usepackage[T1]{fontenc}
\usepackage{lmodern}
\usepackage{graphicx}
\usepackage[shortlabels]{enumitem}
\usepackage{tikz}
\usepackage{pgfmath}

\usetikzlibrary{math}

\usepackage{verbatim}
\usepackage[colorlinks,
            linkcolor=blue,
            anchorcolor=blue,
            citecolor=blue
            ]{hyperref}
\usepackage{cancel}
\usepackage[all]{xy}
\usepackage{mdframed}
\usepackage[capitalize]{cleveref}

\Crefname{figure}{Figure}{Figures}

\newcommand\defeq{\ensuremath{\stackrel{\rm def}{=}}} % Equal by definition

\theoremstyle{plain}
\newtheorem{thm}{Theorem}

\newtheorem{assumption}{Assumption}
\newtheorem{lemma}[thm]{Lemma}

\newtheorem{cor}[thm]{Corollary}
\newtheorem*{cor*}{Corollary}

\theoremstyle{definition}
\newtheorem{defn}[thm]{Definition}

\newtheorem*{rmk}{Remark}

\theoremstyle{remark}

\newcommand\LCS{\textup{LCS}}

\newcommand{\N}{\mathbb{N}}

\newcommand{\F}{\mathbb{F}}

\newcommand{\cO}{\mathcal{O}}

\newcommand{\eps}{\varepsilon}

\newcommand{\range}{\mathrm{Set}}
\newcommand{\var}{\mathrm{Var}}

\newcommand{\floor}[1]{\left\lfloor #1 \right\rfloor}
\newcommand{\ceil}[1]{\left\lceil #1 \right\rceil}

\newcommand{\poly}{\mathsf{poly}}

\newcommand{\rank}{\mathrm{rank}}
\renewcommand{\epsilon}{\varepsilon}

\newcommand{\ed}[2]{\textup{ED}(#1, #2)}
\newcommand{\Fq}{\mathbb{F}_q}

\newif\ifmynotes
\mynotesfalse

\begin{document}
\date{}
 \title{Random Reed--Solomon Codes Achieve the Half-Singleton Bound
 for Insertions and Deletions over Linear-Sized Alphabets}
\author{
Roni Con\thanks{Department of Computer Science, Technion--Israel Institute of Technology. \href{mailto:roni.con93@gmail.com}{\texttt{roni.con93@gmail.com}}  }
\and
Zeyu Guo\thanks{Department of Computer Science and Engineering, The Ohio State University. \href{mailto:zguotcs@gmail.com}{\texttt{zguotcs@gmail.com}}}  
\and
Ray Li\thanks{Mathematics and Computer Science Department, Santa Clara University. \href{mailto:rli6@scu.edu}{\texttt{rli6@scu.edu}}}  
\and
Zihan Zhang\thanks{Department of Computer Science and Engineering, The Ohio State University. \href{mailto:zhang.13691@buckeyemail.osu.edu}{\texttt{zhang.13691@osu.edu}} 
}}
\maketitle
\begin{abstract}
In this paper, we prove that with high probability, random Reed--Solomon codes approach the half-Singleton bound --- the optimal rate versus error tradeoff for linear insdel codes ---  with linear-sized alphabets.
More precisely, we prove that, for any $\epsilon>0$ and positive integers $n$ and $k$, with high probability, random Reed--Solomon codes of length $n$ and dimension $k$ can correct $(1-\varepsilon)n-2k+1$ adversarial insdel errors over alphabets of size $n+2^{\poly(1/\varepsilon)}k$.
This significantly improves upon the alphabet size demonstrated in the work of Con, Shpilka, and Tamo (IEEE TIT, 2023), who showed the existence of Reed--Solomon codes with exponential alphabet size $\widetilde O\left(\binom{n}{2k-1}^2\right)$ precisely achieving the half-Singleton bound.

Our methods are inspired by recent works on list-decoding Reed--Solomon codes. Brakensiek--Gopi--Makam (STOC 2023) showed that random Reed--Solomon codes are list-decodable up to capacity with exponential-sized alphabets, and Guo--Zhang (FOCS 2023) and Alrabiah--Guruswami--Li (STOC 2024) improved the alphabet-size to linear. We achieve a similar alphabet-size reduction by similarly establishing strong bounds on the probability that certain random rectangular matrices are full rank.
To accomplish this in our insdel context, our proof combines the random matrix techniques from list-decoding with structural properties of Longest Common Subsequences.
\end{abstract}
\newpage
\section{Introduction}

Error-correcting codes (hereafter referred to as codes) are constructs designed to enable the recovery of original information from data that has been corrupted. 
The primary corruption model explored in the literature involves errors or erasures. In this model, each symbol in the transmitted word is either replaced with a different symbol from the alphabet (an error) or with a '?' (an erasure).

The theory of coding theory originated with the influential works of Shannon and Hamming. Shannon \cite{shannon1948mathematical} studied random errors and erasures, whereas Hamming \cite{hamming1950error} studied the adversarial model for errors and erasures. 
These models are well understood, and today, we have efficiently encodable and decodable codes that are optimal for Shannon's model of random errors. For adversarial errors, optimal and efficient codes exist over large alphabets, and there are good codes (codes with constant relative rate and relative distance) for every constant-sized alphabet.

Another important model that has been considered ever since Shannon's work is that of \emph{synchronization} errors. The most common model for studying synchronization errors is the insertion-deletion model (insdel for short): an insertion error is when a new symbol is inserted between two symbols of the transmitted word and a deletion is when a symbol is removed from the transmitted word. For example, over the binary alphabet, when $100110$ is transmitted, we may receive the word $1101100$, which is obtained from two insertions ($1$ at the beginning and $0$ at the end) and one deletion (one of the $0$'s at the beginning of the transmitted word). These errors are related to the more common \emph{substitution} and \emph{erasure} errors: a deletion is like an erasure, but the position of the deletion is unknown, and a substitution can be imitated by a deletion followed by an insertion. 
Despite their apparent similarity to well studied error models, insdel errors are much more challenging to handle.

Coding schemes that correct insdel errors are not only an intriguing theoretical concept but also have implications in real-world scenarios. Such codes find applications in magnetic and optical recording, semiconductor devices, integrated circuits, and synchronous digital communication networks (for detailed applications, refer to the survey by Mercier \cite{mercier2010survey}).
This natural theoretical model together coupled with its 
relevance across various domains, has led many researches in recent years to study and design codes that can correct from insdel errors \cite{bukh2016improved,haeupler2021synchronization-org,brakensiek2017efficient,guruswami2017deletion,schoeny2017codes,cheng2018deterministic,haeupler2019optimal,cheng2020efficient,guruswami2021explicit,guruswami2022zero}, just to name a few. 
Although there has been significant progress in recent years on understanding this model of insdel errors (both on limitation and constructing efficient codes), our comprehension of this model lags far behind our understanding of codes that correct erasures and substitution errors (we refer the reader to the following excellent surveys \cite{mitzenmacher2009survey,mercier2010survey,cheraghchi2020overview,haeupler2021synchronization}).
	
    Codes that can correct insdel errors also attract a lot of attention due to their possible application in designing DNA-based storage systems \cite{goldman2013towards}. This recent increased interest was paved by substantial progress in synthesis and sequencing technologies. The main advantages of DNA-based storage over classical
storage technologies are very high data densities and long-term reliability without an electrical supply. It is thus natural that designing codes for DNA storage and studying the limitations of this model received a lot of attention recently \cite{heckel2017fundamental,lenz2019coding,heckel2019characterization,shomorony2022information}.

    While linear codes are highly desirable, most of the works constructing codes for the insdel model are not linear codes, in contrast to the predominant use of linear codes for Hamming errors. Notable examples of linear codes for Hamming errors include Hamming codes, Reed–Solomon codes, Reed–Muller codes, algebraic-geometry codes, polar codes, Turbo codes, expander codes, and LDPC codes.
    The reason for the absence of linear codes in the insdel model can be found in works such as \cite{abdel2007linear, cheng2020efficient}, which show that the maximal rate of a linear code capable of correcting insdel errors is significantly worse than that of a nonlinear code correcting the same number of insdel errors.
    However, linear codes are desirable for many reasons: they offer compact representation, efficient encoding, easier analysis, and in many cases, efficient decoding algorithms.
    Moreover, linear codes are a central mathematical object that has found widespread applications across various research areas. As such, studying them in the context of insdel errors is important and was the subject of recent works \cite{cheng2020efficient,chen2022coordinate,con2022explicit,con2023reed,ji2023strict,cheng2023linear,liu2024optimal}.

    This work concerns the performance of Reed--Solomon codes in the insdel model. As Reed--Solomon codes are ubiquitous in theory and practice, 
    it is important to understand whether they can also decode from insdel errors. This question received a lot of attention recently \cite{safavi2002traitor,wang2004deletion,tonien2007construction,duc2019explicit,liu20212,chen2021improved,con2023reed,con2023optimal,liu2024optimal}. 
    Our result makes a substantial improvement in this line of research. Specifically, we show that ``most'' RS codes over linear-sized fields have almost optimal capabilities correcting insdel errors. 
Our methods are inspired by recent works on list-decoding Reed--Solomon codes \cite{GZ23,alrabiah2023randomly}, which showed that `most' RS codes are list-decodable up to capacity over linear-sized alphabets. Specifically, we achieve our result by 
    adapting the random matrix techniques from \cite{GZ23,alrabiah2023randomly} to the insdel setting which required the development of several chain-based structural results of longest common subsequences.

    \subsection{Previous works}
    \paragraph{Insdel codes.} To offer a wider context, we now discuss some results related to the broader study of codes correcting adversarial insdel errors. 

    Codes correcting adversarial insdel errors were first considered in the seminal works of Levenshtein \cite{levenshtein1966binary} and Varshamov and Tenengolts \cite{varshamov1965codes}, the latter of which constructed binary codes that can correct a single insdel error with optimal redundancy. 
    The first efficient construction of asymptotically good binary insdel codes --- codes with asymptotically positive rate and tolerating an asymptotically positive insdel fraction --- are, as far as we know, given in \cite{schulman1999asymptotically}. 

    Despite these early successes and much recent interest, there are still significant gaps in our understanding of insdel codes, particularly for binary codes.
    For binary codes, gaps remain for several important questions, including (i) determining the optimal redundancy-error tradeoff and finding explicit binary codes for correcting $t$ insdels, for constant $t$ \cite{guruswami2021explicit,sima2020optimal,sima2020systematic}; (ii) finding explicit binary codes for correcting $\varepsilon n$ insdels for small constant $\varepsilon>0$ with optimal rate \cite{cheng2018deterministic,haeupler2019optimal}, (iii) determining the \emph{zero-rate threshold}, the maximum fraction of errors correctable by an asymptotically positive rate code (we know the answer to be in $[\sqrt{2}-1,1/2-10^{-40}]$) \cite{guruswami2016efficiently, guruswami2022zero}, and (iv) determining the optimal rate-versus distance tradeoff for deletion codes \cite{levenshtein1966binary,levenshtein2002bounds,yasunaga2024improved}, among others. 
    
    On the other hand, when the alphabet can be larger, as is the case in this work, the picture is more complete.
    Haeupler and Shahrasbi \cite{haeupler2021synchronization-org}, using a novel primitive called \emph{synchronization strings}, constructed codes of rate $1 - \delta - \varepsilon$ that can correct $\delta$ fraction of insdel errors over an alphabet of size $\cO_{\epsilon}(1)$. 
    This work was the first one to show explicit and efficient constructions of codes that can achieve the Singleton bound in the edit distance setting over constant alphabets.
    This work also inspired several follow-up works that improved the secondary code parameters and applied synchronization strings to related problems; we refer the reader to \cite{haeupler2021synchronization} for a discussion.

    \paragraph{Linear insdel codes.}

    As far as we know, the first study about the performance of linear codes against insdel errors is due to \cite{abdel2007linear}. Specifically, they showed that \emph{any} linear code that can correct one deletion must have rate at most $1/2$. 
    Note that this result shows that linear codes are provably worse than nonlinear codes in the context of insdel errors. Indeed, nonlinear can achieve rate close to $1$ whereas linear codes have rate $\leq 1/2$. 
    
    Nevertheless, as described in the introduction, studying the performance of linear codes against insdel errors is an important question that indeed has been studied in recent years. 
    The basic question of whether there exist good linear codes for the insdel model was first addressed in the work of Cheng, Guruswami, Haeupler, and Li \cite{cheng2020efficient}. They showed that there are linear codes of rate $R = (1-\delta)/2 - h(\delta)/\log_2(q)$ that can correct a $\delta$ fraction of insdel errors. 
    They also established two upper bounds on the rate of linear codes that can correct a $\delta$ fraction of deletions. First, they proved the half-Plotkin bound (\cite[Theorem 5.1]{cheng2020efficient}), which states that every linear code over a finite field $\Fq$ capable of correcting a $\delta$-fraction of insdel errors has rate at most $\frac{1}{2} \left( 1 - \frac{q}{q-1} \delta \right) + o(1)$.
    Then, they established the following alphabet-independent bound, 
    \begin{thm}[{Half-Singleton bound \cite[Corollary 5.2]{cheng2020efficient}}]
		Every linear insdel code which is capable of correcting a $\delta$ fraction of deletions has rate at most $(1-\delta)/2 + o(1)$.
\end{thm}

    \begin{rmk}
        \label{main-remark}
        The following  non-asymptotic version of the half-Singleton bound can be derived from the proof of \cite[Corollary 5.2]{cheng2020efficient}:  
        An $[n,k]$ linear code can correct at most  $n-2k+1$ insdel errors. 
    \end{rmk}
The question of constructing explicit linear codes that are efficiently decodable has also been studied and there are explicit and efficient constructions of linear codes over \emph{small alphabets} (such as binary) that are asymptotically good \cite{con2022explicit,cheng2023linear}. However, we note that over \emph{small alphabets}, determining the optimal rate-error tradeoff and finding optimal explicit constructions remains an open problem.

\paragraph{Reed--Solomon codes against insdel errors.}
In this work, we focus on \emph{Reed--Solomon} codes, which are among the most well-known codes. These codes are defined as follows.
\begin{defn}[Reed--Solomon code]
    Let $\alpha_1, \alpha_2, \ldots, \alpha_n$ be distinct elements in the finite field $\mathbb{F}_q$ of order  $q$. For $k<n$, the $[n,k]_q$ \emph{Reed--Solomon} (RS) code of dimension $k$ and block length $n$ associated with the evaluation vector $(\alpha_1, \ldots, \alpha_n)\in\F_q^n$ is defined to be the set of codewords 
    \[
    \text{RS}_{n,k}(\alpha_1, \ldots, \alpha_n) := \left \lbrace \left( f(\alpha_1), \ldots, f(\alpha_n) \right) \mid f\in \mathbb{F}_q[x],\text{ }\deg f < k \right \rbrace.
    \]	
\end{defn}

The performance of RS codes against insdel errors was first considered in \cite{safavi2002traitor} in the context of traitor tracing. In \cite{wang2004deletion}, the authors constructed a $[5,2]_q$ RS code correcting a single deletion. In \cite{tonien2007construction}, an $[n,k]$ generalized RS code that is capable of correcting $\log_{k+1}n - 1$ deletions was constructed. 
Several constructions of two-dimensional RS codes correcting $n-3$ insdel errors are given in \cite{duc2019explicit,liu20212}. Note that the half-Singleton bound states that an $[n,2]$ code can correct at most $n-3$ deletions. In \cite{con2023reed}, it was shown that for an $[n,2]_q$ RS code to decode from $n-3$ insdel errors, it must be that $q=\Omega(n^3)$. On the other hand, \cite{con2023optimal} gave an explicit construction with $q=O(n^3)$. Thus, for the special case of two-dimensional RS codes, there is a complete characterization of RS codes that achieve the half-Singleton bound.

For $k>2$, much less is known. It was shown in \cite{con2023reed} that over large enough fields, there exist RS codes that exactly achieve the half-Singleton bound. 
Specifically, 
\begin{thm}[{\cite[Theorem 16]{con2023reed}}] \label{thm:con-rs-res}
    Let $n$ and $k$ be positive integers such that $2k - 1 \leq n$. For $q = 2^{\Theta(n)}$, there exists an $[n,k]_q$ RS code that can decode from $n-2k+1$ insdel errors.\footnote{The alphabet size in \cite[proof of Theorem 16]{con2023reed} is actually $q=\binom{n}{2k-1}^2 \cdot k^2 + n^2$, which is better, especially for sublinear $k=o(n)$.
    However, given that the primary parameter regime of interest in this paper is $k=\Theta(n)$, we state this simplified version for brevity.}
\end{thm}
In both theoretical and practical scenarios, codes over smaller alphabets tend to be more valuable, but the field size above is very large. Whether Reed--Solomon codes over significantly smaller fields are capable of correcting insdel errors up to the half-Singleton bound remained an interesting and important open problem, which we address in this work.

\paragraph{List-decodable (Hamming) codes.}
As we  mentioned, we port some crucial ideas from a different context (List-decoding under Hamming distance) to build our main results in the insertion-deletion world. 
To facilitate a better understanding, we provide a concise summary of recent and significant advancements in the field of list decoding of codes under the Hamming metric. 

The notion of \emph{list decoding} was introduced by Elias and Wozencraft \cite{elias1957list,wozencraft1958list}
in 1950s as a natural generalization of unique decoding. 
Briefly, a code exhibits satisfactory list decodability if its codewords are widely dispersed throughout the entire space, that is, there are not too many codewords located within a single ball under the Hamming metric. After the seminal work of Sudan \cite{sudan1997decoding} and Guruswami--Sudan \cite{guruswami1998improved}, which provided efficient algorithms for list decoding RS codes up to the Johnson bound \cite{johnson1962new}, the issue of understanding the list decodability of RS codes beyond the Johnson bound became very important. A long line 
of works \cite{rudra2014every, shangguan2020combinatorial, guo2022improved, ferber2022list, goldberg2022list,brakensiek2023generic,GZ23, alrabiah2023randomly, brakensiek2024ag, ron2024efficient, guo2024random} 
have made significant advancements in the understanding of list decodability for RS and related codes. Specifically, recent works \cite{brakensiek2023generic, GZ23, alrabiah2023randomly} have shown that ``most'' RS codes over exponential-sized alphabets (in terms of the code length) are optimally list decodable, and ``most'' RS codes over linear-sized alphabets are in fact almost optimally list decodable.
\subsection{Our results}

When $(\alpha_1, \ldots, \alpha_n)$ is uniformly distributed over the set of all $n$-tuples of distinct elements in $\F_q$, we say the code $\text{RS}_{n,k}(\alpha_1, \ldots, \alpha_n)$ over $\Fq$ is a \emph{random RS code} of dimension $k$ and length $n$ over $\Fq$.
In this work, we 
show that random RS codes over alphabets of size $O_{\varepsilon}(n)$, with high probability, approach the half-Singleton bound for insdel errors.
Specifically, 
\begin{thm}[Informal, Details in Theorem~\ref{main}]\label{main2}
Let $\varepsilon\in (0,1)$, and let $n$ and $k$ be positive integers. 
For a prime power $q$ satisfying $q\ge n+ 2^{\mathrm{poly}(1/\epsilon)}k$, with high probability, a random RS code of dimension $k$ and length $n$ over $\Fq$ corrects at least 
$ (1-\epsilon)n-2k+1$ adversarial insdel errors.
\end{thm}

    For the constant rate regime, $R=\Theta(1)$, our result exponentially improves the alphabet size of Con, Shpilka, and Tamo \cite{con2023reed}, where they have $q=2^{\Theta(n)}$.
    As a warmup to this result, we prove a weaker but more straightforward result (\cref{m1}), which establishes Theorem~\ref{main2} for $q=2^{O(1/\varepsilon)}n^2$.

\subsection{Proof overview}
\label{sec:proof-overview}
We outline the proof of our main theorem in this section. 
First, we review the proof of Theorem~\ref{thm:con-rs-res} \cite{con2023reed} that achieves the half-Singleton bound with exponential-sized alphabets. 
We slightly modify the proof's presentation to parallel our proofs.
Second, we show how to prove a weaker version of Theorem~\ref{main2} (Theorem~\ref{m1}) with quadratic-sized alphabets.
Lastly, we describe how to prove Theorem~\ref{main2} that achieves linear-sized alphabets.
Throughout this overview, let $C$ be a random Reed--Solomon code of length $n$ and dimension $k$, where the tuple of evaluation points $(\alpha_1,\dots,\alpha_n)$ is sampled uniformly from all $n$-tuples of pairwise-distinct elements from $\mathbb{F}_q$.

\paragraph{Warmup: exponential alphabet size.} 
We start with the argument from \cite{con2023reed} that proves that Reed--Solomon codes achieve the half-Singleton bound over exponential-sized alphabets $q=2^{\Theta(n)}$.
Let $\ell\defeq 2k-1$.
We want (with high probability over $\alpha_1,\dots,\alpha_n$) that our code $C$ corrects $n-2k+1=n-\ell$ insdel errors, or equivalently, has pairwise small LCS: $\LCS(c,c') < \ell$ for any two distinct codewords $c,c'\in C$.

The key observation is that, if our code $C$ fails to correct $n-2k+1$ insdels, then there exist indices $I_1<\cdots<I_\ell$ and $J_1<\cdots<J_\ell$ such that the following matrix, which we call the \emph{$V$-matrix}, is \emph{bad}, meaning it does not have full column rank.
\begin{align}
    V_{k,\ell,I,J}(\alpha_1,\alpha_2,\cdots,\alpha_n):=\left(\begin{array}{ccccccc}1 & \alpha_{I_1} & \cdots & \alpha_{I_1}^{k-1} & \alpha_{J_1} & \cdots & \alpha_{J_1}^{k-1} \\1 & \alpha_{I_2} & \cdots & \alpha_{I_2}^{k-1} & \alpha_{J_2} & \cdots & \alpha_{J_2}^{k-1} \\ \vdots& \vdots & \ddots& \vdots & \vdots & \ddots & \vdots \\1 & \alpha_{I_\ell} & \cdots & \alpha_{I_\ell}^{k-1} & \alpha_{J_\ell} & \cdots & \alpha_{J_\ell}^{k-1}\end{array}\right),
\end{align}
Indeed, if $C$ fails to correct $n-2k+1$ insdels, there exist two distinct polynomials $f(X)=f_0+f_1X+\cdots+f_{k-1}X^{k-1}$ and $f'(X) = f_0'+f_1'X+\cdots+f'_{k-1}X^{k-1}$, such that the $(I_1,\dots,I_\ell)$-indexed subsequence of the codeword for $f$ equals the $(J_1,\dots,J_\ell)$-indexed subsequence of the codeword for $f'$. In that case, 
\begin{align}
\left(\begin{array}{ccccccc}1 & \alpha_{I_1} & \cdots & \alpha_{I_1}^{k-1} & \alpha_{J_1} & \cdots & \alpha_{J_1}^{k-1} \\1 & \alpha_{I_2} & \cdots & \alpha_{I_2}^{k-1} & \alpha_{J_2} & \cdots & \alpha_{J_2}^{k-1} \\ \vdots& \vdots & \ddots& \vdots & \vdots & \ddots & \vdots \\1 & \alpha_{I_\ell} & \cdots & \alpha_{I_\ell}^{k-1} & \alpha_{J_\ell} & \cdots & \alpha_{J_\ell}^{k-1}\end{array}\right)
    \cdot\begin{pmatrix}
        f_0-f_0'\\
        f_1\\
        \vdots\\
        f_{k-1}\\
        -f_1'\\
        \vdots\\
        -f_{k-1}'\\
    \end{pmatrix}
    =0,
\end{align}
so the $V$-matrix is bad.

Now, it suffices to show, with high probability, that all $V$-matrices are good (have full collumn rank).
However, by considering the determinant of the $V$-matrix (which is square as $\ell=2k-1$), the probability that one $V$-matrix is bad is at most $\frac{k(k-1)}{q-n}$ by the Schwartz--Zippel lemma.\footnote{An important detail here, which \cite{con2023reed} proves, is that (to apply the Schwartz-Zippel lemma) the determinant needs to be symbolically nonzero.} 
A $V$-matrix is defined by the indices of the subsequences $I_1<\cdots<I_\ell$ and $J_1<\cdots<J_\ell$, so there are at most $2^{2n}$ possible $V$-matrices.
Hence, by the union bound, the probability that some $V$-matrix is bad is at most $2^{2n}\cdot \frac{k(k-1)}{q-n}$.
Hence, for sufficiently large exponential alphabet sizes $q\ge 2^{\Theta(n)}$, our code corrects $n-2k+1$ insdel errors with high probability, as desired.

\paragraph{Quadratic alphabet  size.} 
We now discuss how to improve the field size bound, first to quadratic, and then to linear.

Our main idea, inspired by \cite{GZ23,alrabiah2023randomly}, is to use ``slackness'' in the coding parameters to amplify the probability that a $V$-matrix is bad.
The above warmup gives random RS codes that correct $n-(2k-1)$ errors, \emph{exactly} achieving the half-Singleton bound.
We relax the guarantee, and now ask for a random RS code to correct $n-(2k-1)-\varepsilon n$ errors, \emph{approaching} the half-Singleton bound.
Now, the corresponding $V$-matrix is a $\ell\times (2k-1)$ matrix, for $\ell\defeq (2k-1)+\varepsilon n$.
For this relaxation, we show the probability $V$-matrix is bad is at most $\left(\frac{kn}{q-n}\right)^{\Theta(\varepsilon n)}$, rather than $\frac{k(k-1)}{q-n}$. 

First, we discuss a toy problem that illustrates this probability amplification.
Consider the toy problem of independently picking $\ell$ \emph{uniformly} random row vectors $v_1,\cdots,v_\ell\in\mathbb{F}_q^{2k-1}$ to form an $\ell\times (2k-1)$ matrix $M$, which we want to have full column rank. If we choose $\ell=2k-1$, then the probability that $M$ has full column rank is bounded by a function that is $\Theta(1/q)$, and this happens only if each $v_i$ is not in the span of $v_1,\dots, v_{i-1}$. However, suppose we choose $\ell = (2k-1)+\varepsilon n$ for some small $\varepsilon > 0$. In this case, we could afford $\varepsilon n$ ``faulty'' vectors $v_i$ , i.e., $v_i$ may be in the span of previous vectors, in which case we just skip it and consider the next vector. The probability that the matrix $M$ has full column rank is then exponentially small, $1/q^{\Omega(\varepsilon n)}$.

Now we outline the formal proof of this probability amplification, captured in Corollary~\ref{cor:fullrankprob}.
\begin{cor*}[Corollary~\ref{cor:fullrankprob}, informal]
Let $\varepsilon\in[0,1]$, $\ell=(2k-1)+\varepsilon n$, and $r = \varepsilon n/2$. Let $I,J\in[n]^{\ell}$ be two increasing subsequences that agree on at most $k-1$ coordinates\footnote{More precisely, $I_i=J_i$ for at most $k-1$ values of $i$. This technical condition ensures that the $V$-matrix is symbolically full column rank.}. Then,
\[
\Pr\bigg[\text{Matrix $V_{k,\ell,I,J}(\alpha_1,\dots,\alpha_n)$ is bad}\bigg]\leq 
\left(\frac{2n(k-1)}{q-n+1}\right)^r.
\]
\end{cor*}

At the highest level, our proof of Corollary~\ref{cor:fullrankprob}  is a union bound over ``certificates.''
For all evaluation points $(\alpha_1,\dots,\alpha_n) \in \mathbb{F}_q^n$ where $V_{k,\ell,I,J}$ is bad, we show that there is a \emph{certificate} $(i_1,\dots,i_r) \in [n]^r$ of distinct indices in $[n]$ (\Cref{lem:twocases}) that intuitively ``attests'' that the matrix $V_{k,\ell,I,J}$ is bad.

We generate the certificate $(i_1,\dots,i_r)$ deterministically from the evaluations $\alpha_1,\dots,\alpha_n$ using Algorithm~\ref{alg:algorithm}.
We compute the certificate one index $i_j$ at a time.
Given indices $i_1,\dots,i_{j-1}$, define index $i_j$ as follows: let $A_j$ be the top $(2k-1)\times (2k-1)$ square submatrix of $V_{k,\ell,I,J}^{\{i_1,\dots,i_{j-1}\}}$ --- the $V$-matrix $V_{k,\ell,I,J}$ after removing rows containing any of variables $X_{i_1},X_{i_2},\dots,X_{i_{j-1}}$ --- and let $i_j$ be the smallest index such that $A_j|_{X_1=\alpha_1,\dots,X_{i_j}=\alpha_{i_j}}$ is not full column rank (we call $i_j$ a \emph{faulty index}, Definition~\ref{defn:faulty}).
Since $A_j$ is a full rank submatrix of a bad $V$-matrix\footnote{The full-rank part needs to be checked, but follows from \cite{con2023reed}.}, $A_j|_{X_1=\alpha_1,\dots,X_n=\alpha_n}$ is not full rank, so index $i_j$ always exists. 
Hence, we can keep generating indices $i_j$ as long as the truncated $V$-matrix, $V_{k,\ell,I,J}^{\{i_1,\dots,i_{j-1}\}}$, has at least $2k-1$ rows.
By definition, each $X_i$ participates in at most 2 rows of a $V$-matrix, so we get a certificate of length at least $\floor{\frac{\ell-(2k-1)}{2}}+1 \ge \varepsilon n / 2 = r$.

We then show (Lemma~\ref{lem:certprob}), for any certificate $(i_1,\dots,i_r)$, the probability that the $V$-matrix has certificate $(i_1,\dots,i_r)$ is exponentially small.
Conditioned on $A_j$ being full rank with $X_1=\alpha_1,\dots,X_{i_j-1}=\alpha_{i_j-1}$, the probability that $A_j$ becomes not-full-rank when setting $X_{i_j}=\alpha_{i_j}$ is at most $\frac{2(k-1)}{q-n}$: $\alpha_{i_j}$ is uniformly random over at least $q-n$ field elements, and the degree of $X_{i_j}$ in the determinant of $A_j$ is at most $2(k-1)$.
This event needs to happen $r$ times, and it is possible to run the conditional probabilities in the correct order to conclude that the probability of generating a particular certificate $(i_1,\dots,i_r)$ is at most $\left(\frac{2(k-1)}{q-n}\right)^r$.

Since there are at most $n^r$ certificates, the total probability that that a particular $V$-matrix is bad is at most $n^r\cdot \left(\frac{2(k-1)}{q-n}\right)^r$. This is at most $2^{-3n}$ for sufficiently large quadratic alphabet sizes $q=2^{\Theta(1/\varepsilon)}\cdot n^2$. 
For such $q$, by a union bound over the at-most-$2^{2n}$ $V$-matrices, with probability at least $1-2^{-n}$, the code $C$ corrects $n-\ell$ deletions, thus approaching the half-Singleton bound with quadratic alphabet size.

\paragraph{Linear alphabet size.} 
To improve the alphabet size to linear, we modify the certificate argument so that the number of certificates is only $\binom{n}{r}$, rather than $n^r$.
The idea is to force the certificates to have increasing coordinates $i_1<i_2<\cdots<i_r$ (this does not hold automatically).\footnote{For convenience, in the actual proof, our certificate is slightly different. Instead of $1\le i_1<i_2<\cdots<i_r\le n$, we take certificates $0\le i_1\le i_2\le\cdots\le i_r\le 2k-2$. The $i_j$'s have slightly different  meaning, but the idea is the same.} 

First, we preprocess the $V$-matrix by removing some of its rows (equivalently, we remove elements from $I$ and $J$), so that the remaining matrix can be partitioned into ``blocks'' of length at most $O(1/\varepsilon)$.
Crucially, the variables in a block appear \emph{only} in that block, so that the blocks partition the variables $X_1,\dots,X_n$ (in the proof, these blocks are given by \emph{chains}, Definition~\ref{defn:chain}).
Proving this step requires establishing structural properties of longest common subsequences.

We then generate our certificates in a more careful way.
We remove the largest $\Omega_\varepsilon(n)$ of these blocks from our $V$-matrix to create a \emph{bank} of blocks, and, we reindex the variables so that the banked blocks have the highest-indexed variables.\footnote{This is acceptable, since we don't use the original ordering on the variables after this point.}
As in Algorithm~\ref{alg:algorithm}, we choose $A_1$ to be the top $(2k-1)\times (2k-1)$ submatrix of the $V$-matrix --- this time, after removing the blocks in the bank ---, and for all $j$, we again choose $i_j$ as the smallest index such that setting $X_{i_j}=\alpha_{i_j}$ makes $A_j$ not full rank.
However, we choose the matrices $A_2,A_3,\dots$ more carefully.
After choosing $i_j$, we let $A_{j+1}$ be a submatrix of $V_{k,\ell,I,J}$ that ``re-indeterminates'' the matrix $A_j$: we remove from $A_j$ the block containing variable $X_{i_j}$, and replace it with an ``equivalent'' new block from our bank --- possibly truncating the new block, if the new block is longer than the old block --- to get $A_{j+1}$. This results in a matrix $A_{j+1}$ ``equivalent'' to $A_j$; it is the same polynomial matrix up to permuting the rows and relabeling the indeterminates.
Since this matrix $A_{j+1}$ is an equivalent, ``re-indeterminated'' version of $A_j$, we must have $A_{j+1}|_{X_1=\alpha_1,\dots,X_{i_j}=\alpha_{i_j}}$ is full column rank, so we have $i_j < i_{j+1}$, which is our desired property for our certificates. Further, since our bank has at least $\Omega_\varepsilon(n)$ blocks, we can ``re-indeterminate'' at least $\Omega_\varepsilon(n)$ times, giving a certificate of length $r=\Omega_\varepsilon(n)$.

Since our certificates now satisfy $i_1<\cdots<i_r$, the number of certificates is at most $\binom{n}{r}$, the probability $C$ fails to correct $n-\ell$ deletions is only $\binom{n}{r}\left(\frac{2(k-1)}{q-n}\right)^r$, which is exponentially small $2^{-3n}$ for sufficiently large linear alphabet sizes $q=n+\Theta_\varepsilon(k)$.
Again, a union bound over (at most $2^{2n}$) $V$-matrices gives that, with high probability, our code $C$ corrects $n-\ell$ deletions, thus approaching the half-Singleton bound with linear-sized alphabets.

\subsection{Future research directions}

We conclude the introduction by outlining several open questions for future research:

\begin{enumerate}
    \item  \emph{Lower bounds on the field size.} In \cite{con2023reed}, it was demonstrated that there exist RS codes over fields of size $n^{O(k)}$ that exactly attain the half-Singleton bound. This paper shows that the field size can be significantly reduced if we only need to get $\varepsilon$-close to the half-Singleton bound. A remaining open question is to prove a lower bound on the field size of linear codes, not just RS codes, that achieve (or get close to) the half-Singleton bound.

    \item \emph{Explicit constructions.} The results presented in this paper are existential; specifically, we demonstrate the existence of RS codes capable of correcting insdel errors with an almost optimal rate-distance tradeoff. An important question remains: How to provide explicit constructions of RS codes that achieve these parameters?
    
    \item \emph{Decoding algorithms.} One of the primary advantages of RS codes in the Hamming metric is their efficient decoding algorithms. Currently, as far as we know, there are no efficient algorithms for RS codes that correct insdel errors. Therefore, a crucial open question is how to design efficient decoding algorithms for RS codes handling insdel errors.

     \item \emph{Affine codes.} In \cite[Theorem 1.5]{cheng2020efficient}, it was demonstrated that affine codes outperform linear codes in correcting insdel errors. Specifically, efficient affine binary codes of rate $1-\varepsilon$ were constructed to correct $\Omega(\varepsilon^3)$ insdel errors. By efficient codes, we mean explicit construction of codes with efficient encoding and decoding algorithms. An immediate open question is to construct efficient affine code with improved rate-distance trade-offs. Moreover, considering that RS codes achieve the half-Singleton bound, a natural question arises: can affine RS codes perform even better? In particular, can they approach the Singleton bound? 
\end{enumerate}
\subsection*{Acknowledgments}
This work was conducted while all four authors were visiting the Simons Institute for the Theory of Computing at UC Berkeley. The authors would like to thank the institute for its support and hospitality. Additionally, Z. Zhang wishes to thank Prof. Xin Li for the helpful discussions and suggestions. R. Li is supported by NSF grant CCF-2347371.

\section{Notation and Preliminaries}
Define $\N=\{0,1,2,\dots,\}$, $\N^+=\{1,2,\dots\}$, and $[n]=\{1,2,\ldots, n\}$ for $n\in\N$. Throughout the paper, $\Fq$ denotes the finite field of order $q$.

Denote by $\Fq[X_1, \ldots, X_n]$ the ring of polynomials in $X_1,\dots,X_n$ over $\Fq$ 
and by $\Fq(X_1, \ldots, X_n)$ the field of fractions of that ring.
More generally, for a collection of variables $X_i$ indexed by a set $I$, denote by $\Fq[X_i: i\in I]$ the ring of polynomials in these variables over $\Fq$.
The value of $f\in \Fq[X_i: i\in I]$ at $(\alpha_i)_{i\in I}$ is denoted by $f(\alpha_i: i\in I)$, or $f(\alpha_1,\dots,\alpha_n)$ when $I=[n]$.

The degree $\deg(f)$ of a nonzero multivariate polynomial $f\in \Fq[X_i: i\in I]$ refers to its total degree. And for $i\in [n]$, $\deg_{X_i}(f)$ denotes the degree of $f$ in $X_i$, where the other variables are viewed as constants. 

We also need the following notation, adapted from \cite{GZ23}, regarding the partial assignment of a symbolic matrix.

\begin{defn}[Partial assignment]
Let $A$ be a matrix over $\mathbb{K}:=\F_q(X_1,\dots,X_n)$ such that the entries of $A$ are in $\F_q[X_1,\dots,X_n]$.
For $i\in\{0,1,\dots,n\}$ and $\alpha_1,\dots,\alpha_i\in \mathbb{K}$, denote by $A|_{X_{1}=\alpha_{1},\dots,X_{i}=\alpha_{i}}$ the matrix obtained from $A$ by substituting $\alpha_{j}$ for $X_{j}$ for $j=1,\dots,i$. 
More generally, for $I\subseteq [n]$ and a tuple $(\alpha_i)_{i\in I}\in\mathbb{K}^I$, denote by $A|_{X_{i}=\alpha_{i} \text{ for } i\in I}$ the matrix obtained from $A$ by substituting $\alpha_{i}$ for $X_{i}$ for $i\in I$. 
\end{defn}

We need the following variation of the Schwarz-Zippel Lemma.

\begin{lemma}\label{lem:zero-prob}
Let $Q(X_1,\dots,X_n)\in\Fq[X_1,\dots,X_n]$ be a nonzero polynomial such that $\deg_{X_i}(Q)\leq d$ for $i\in [n]$.
Let $T\subseteq\Fq$ be a set of size at least $n$, and let $\alpha=(\alpha_1,\dots,\alpha_n)$ be uniformly distributed over the set of $n$-tuples with distinct coordinates in $T$.
Then $\Pr[Q(\alpha_1,\dots,\alpha_n)=0]\leq \frac{nd}{|T|-n+1}$.
\end{lemma}

\begin{proof}
For $i\in [n]$, let $E_i$ be the event that 
$Q(\alpha_1,\dots,\alpha_{i-1},X_i,\dots,X_n)\neq 0$ but
$Q(\alpha_1,\dots,\alpha_i,X_{i+1},\dots,X_n)=0$.
As $Q\neq 0$, the event $Q(\alpha_1,\dots,\alpha_n)=0$ occurs precisely when one of $E_1,\dots,E_n$ occurs. By the union bound, it suffices to prove that $\Pr[E_i]\leq \frac{d}{|T|-n+1}$.

Fix $i\in [n]$. Further fix $\alpha_1,\dots,\alpha_{i-1}\in\Fq$ such that $Q(\alpha_1,\dots,\alpha_{i-1},X_i,\dots,X_n)\neq 0$.
It suffices to prove that $E_i$ occurs with probability at most $\frac{d}{|T|-n+1}$ after fixing such $\alpha_1,\dots,\alpha_{i-1}$.
(If such $\alpha_1,\dots,\alpha_{i-1}$ do not exist, then $\Pr[E_i]=0$ and we are done.)

Now $\alpha_i$ is uniformly distributed over the set $T':=T\setminus\{\alpha_1,\dots,\alpha_{i-1}\}$.
Let $Q^*:=Q(\alpha_1,\dots,\alpha_{i-1},X_i,\dots,X_n)\neq 0$.
View $Q^*$ as a polynomial in $X_{i+1},\dots,X_n$ over the ring $\Fq[X_i]$. Let $Q^*_0\in\Fq[X_i]$ be the coefficient of a nonzero term of $Q^*$.
Then $Q^*_0\neq 0$, and $\deg(Q^*)\leq \deg_{X_i}(Q)\leq d$.
Therefore,
\begin{equation}\label{eq:prob-bound}
\Pr[Q^*_0(\alpha_i)=0]\leq \frac{d}{|T'|}\leq \frac{d}{|T|-n+1}.
\end{equation}
Note that $Q^*_0(\alpha_i)$ is the coefficient of a term of $Q^*(\alpha_i,X_{i+1},\dots,X_n)=Q(\alpha_1,\dots,\alpha_i,X_{i+1},\dots,X_n)$.
Therefore, the latter is zero only if $Q_0^*(\alpha_i)=0$. So by \eqref{eq:prob-bound}, $E_i$ occurs with probability at most $\frac{d}{|T|-n+1}$, as desired.
\end{proof}

We recall the notion of a subsequence and a longest common subsequence. 
\begin{defn}
        A \emph{subsequence} of a string $s$ is a string obtained by removing some (possibly none) of the symbols in $s$. 
\end{defn}
	\begin{defn}
	
		Let $s,s'$ be strings over an alphabet $\Sigma$. 
		A \emph{longest common subsequence} between $s$ and $s'$, is a subsequence of both $s$ and $s'$, of maximal length. We denote by $ \textup{LCS}(s,s')$ the length of a longest common subsequence.
		
	The \emph{edit distance} between $s$ and $s'$, denoted by $\ed{s}{s'}$, is the minimal number of insertions and deletions needed in order to turn $s$ into $s'$. 
\end{defn}

It is well known that the insdel correction capability of a code is determined by the LCS of its codewords. Specifically,
\begin{lemma} \label{lem:code-lcs}
    A code $C$ can correct $\delta n$ insdel errors if and only $\textup{LCS}(c,c')\leq n - \delta n - 1$ for any distinct $c,c'\in C$.
\end{lemma}
We give the proof for completeness.
\begin{proof}
    A code can correct $\delta n$ insdel errors if and only if
    for any two distinct codewords $c,c'$, it holds $\ed{c}{c'} \geq 2\delta n + 1$. 
    The claim follow by the well-known relation between edit distance and LCS which states that
    $\ed{s}{s'} = |s| + |s'| - 2 \textup{LCS}(s,s')$ \cite[Lemma 12.1]{crochemore2003jewels}.
\end{proof}

We now adopt two definitions and a lemma from \cite{con2023reed}. These will establish the algebraic conditions ensuring that an RS code can correct insdel errors.

\begin{defn}[Increasing subsequence]
    We call $I=(I_1,\dots,I_\ell)\in[n]^\ell$ an \emph{increasing subsequence} if it holds that $ I_1<I_2<\cdots<I_\ell$, where $\ell$ is called the \emph{length} of $I$.
\end{defn}

\begin{defn}[$V$-matrix]\label{matrix} 
For positive integers $\ell, k$ and increasing subsequences $I=(I_1,\dots,I_\ell),J=(J_1,\dots,J_\ell)\in [n]^{\ell}$ of length $\ell$, define the $\ell\times(2k-1)$ matrix
\[
V_{k,\ell,I,J}(X_1,X_2,\cdots,X_n):=\left(\begin{array}{ccccccc}1 & X_{I_1} & \cdots & X_{I_1}^{k-1} & X_{J_1} & \cdots & X_{J_1}^{k-1} \\1 & X_{I_2} & \cdots & X_{I_2}^{k-1} & X_{J_2} & \cdots & X_{J_2}^{k-1} \\ \vdots& \vdots & \ddots& \vdots & \vdots & \ddots & \vdots \\1 & X_{I_\ell} & \cdots & X_{I_\ell}^{k-1} & X_{J_\ell} & \cdots & X_{J_\ell}^{k-1}\end{array}\right),
\]
over the field 
$\F_q(X_1,\dots,X_n)$. We call $V_{k,\ell,I,J}(X_1,X_2,\cdots,X_n)$
a \emph{$V$-matrix}. 
\end{defn}

The following lemma states that if a Reed--Solomon code $\mathsf{RS}_{n,k}(\alpha_1, \ldots, \alpha_n)\subseteq\F_q^n$ cannot correct $\ell$ insdel errors, then we can identify a specific $V$-matrix that does not have full column rank. 
The proof of this lemma is identical to the proof of Proposition 2.1 in \cite{con2023reed}. However, we include it here for completeness.

\begin{lemma}\label{bad}
Let $\ell\ge 2k-1$. 
Consider the $[n, k]$ Reed--Solomon code $\mathsf{RS}_{n,k}(\alpha_1, \ldots, \alpha_n)\subseteq\F_q^n$ associated with an evaluation vector $\vec{\alpha}:=\left(\alpha_1, \ldots, \alpha_n\right)\in\F_q^n$. If 
the code cannot correct arbitrary $n-\ell$ insdel errors, then there exist two increasing subsequences $I=(I_1,\dots,I_\ell), J=(J_1,\dots,I_\ell) \in[n]^{\ell}$ that agree on at most $k-1$ coordinates
such that matrix $V_{k,\ell,I, J}|_{X_{1}=\alpha_{1},\dots,X_{n}=\alpha_{n}}$ does not have full column rank.
\end{lemma}

\begin{proof}
Suppose the code $\mathsf{RS}_{n,k}(\alpha_1, \ldots, \alpha_n)\subseteq\F_q^n$ can not recover from $n-\ell$ insdel errors, then by \Cref{lem:code-lcs}, there exist two distinct codewords $c_i:=(f_i(\alpha_1),\dots,f_i(\alpha_n))$ ($i=1,2$) such that $\mathsf{LCS}(c_1,c_2)\ge \ell$. Therefore, by definition, we have two increasing subsequences $I,J\in[n]^{\ell}$ such that
\begin{equation}\label{f1f2}
    f_1(\alpha_{I_s})=f_2(\alpha_{J_s}) \quad \text{for } s=1,2,\dots,\ell,
\end{equation}
where $f_i:=\sum_{j=0}^{k-1}f_i^{(j)}X^j\in\F_q[X]$ ($i=1,2$) are two distinct polynomials of degree at most $k-1$. To verify that $I$ and $J$ agree on at most $k-1$ coordinates, assume that this is not true, i.e., there exist distinct $i_1,\dots,i_k\in [\ell]$ such that $I_{i_j}=J_{i_j}$ for $j\in [k]$. Then by \eqref{f1f2}, $f_1-f_2$ vanishes at the $k$ distinct points $\alpha_{I_{i_1}},\dots,\alpha_{I_{i_k}}$.
But this contradicts the fact that $f_1$ and $f_2$ are distinct polynomials of degree at most $k-1$.

Consider the vector $\vec{v}:=\left(f_1^{(0)}-f_2^{(0)},f_1^{(1)},f_1^{(2)},\dots,f_{1}^{(k-1)},-f_2^{(1)},-f_2^{(2)},\dots,-f_{2}^{(k-1)}\right)\in\F_q^{2k-1}$, which is non-zero. 
By the algebraic relations stated in $(\ref{f1f2})$, we have $V_{k,\ell,I, J}|_{X_{1}=\alpha_{1},\dots,X_{n}=\alpha_{n}}\cdot \vec{v}=0$. 
Therefore, the matrix $V_{k,\ell,I, J}|_{X_{1}=\alpha_{1},\dots,X_{n}=\alpha_{n}}$ does not have full column rank.
\end{proof}

A key technical lemma regarding the $V$-matrices that was proved in \cite{con2023reed} is the following.
\begin{lemma}[Proposition 2.4, \cite{con2023reed}]\label{nonzero}
    Let $I, J \in[n]^{2k-1}$ be two increasing subsequences that agree on at most $k-1$ coordinates.  
		Then, we have $\det\left(V_{k,2k-1,I,J}(X_1,X_2,\ldots,X_n)\right)\neq 0$ as a multivariate polynomial in $\F_q[X_1,X_2,\ldots,X_n]$.
\end{lemma}

Informally speaking, by using this lemma, and by considering all the $V$-matrices, the authors of \cite{con2023reed} showed by using the Schwarz-Zippel lemma that there exists an assignment $(\alpha_1, \ldots, \alpha_n)$ over a large enough field for which $\det (V_{k,\ell,I, J}|_{X_{1}=\alpha_{1},\dots,X_{n}=\alpha_{n}})$ is nonzero for all pairs $I,J$ of length $2k-1$ that agree on at most $k-1$ coordinates.

For convenience, we introduce the following definition. 
\begin{defn}[Selecting a subset of coordinates] \label{defn:restriction}
For an increasing subsequence $I\in [n]^{\ell}$ and $P=\{i_1,\dots,i_{\ell'}\}\subseteq [\ell]$ with $i_1<\dots<i_{\ell'}$, denote by $I^P$ the increasing subsequence $(I_{i_1},\dots,I_{i_{\ell'}})$.
\end{defn}

The following is a corollary of Lemma~\ref{nonzero}.
\begin{cor}
\label{roub}
    Let $I, J \in[n]^{\ell}$ with $\ell\ge 2k-1$ be two increasing subsequences that agree on at most $k-1$ coordinates. Let $P\subseteq [\ell]$ be a set of size $2k-1$. Then $\det (V_{k,2k-1,I^P,J^P})\neq 0$. 
    Additionally, for $i\in[n]$, we have $\deg_{X_i}(\det (V_{k,2k-1,I^P,J^P}))\leq 2(k-1)$. 
\end{cor}
\begin{proof}
Note that $I^P,J^P \in[n]^{2k-1}$ are two increasing subsequences that agree on at most $k-1$ coordinates. So the first claim follows immediately from \cref{nonzero}.
The second claim follows from the facts that every entry of $ V_{k,2k-1,I^P,J^P}$ has the form $X_i^j$ with $j\leq k-1$ and that each $X_i$ appears in at most two rows. 
\end{proof}
\section{Achieving Quadratic Alphabet Size}
In this section, we present a warm-up result, \cref{m1}, which achieves the alphabet size $O_{\eps}(n^2)$. Our proof resembles the proof by Guo and Zhang \cite{GZ23} showing that random RS codes achieve list decodability over fields of quadratic size.

\subsection{Full rankness of $V_{k,\ell,I,J}$ under a random assignment}
Firstly, we introduce the following definition, $V_{k,\ell,I,J}^{B}$, which is a submatrix of $V_{k,\ell,I,J}$ obtained by deleting certain rows. 
\begin{defn}\label{delete_B}
Under the notation in Definition \ref{matrix}, for a subset $B\subseteq [n]$, define the matrix $V_{k,\ell,I,J}^{B}$ to be the submatrix of $V_{k,\ell,I,J}$ obtained by deleting the $i$-th row for all $i\in [\ell]$ such that $I_i\in B$ or $J_i\in B$. 
In other words, $V_{k,\ell,I,J}^{B}$ is obtained from $V_{k,\ell,I,J}$ by deleting all the rows containing powers of $X_j$ for every index $j\in B$.
\end{defn}

We also need to define the notion of \emph{faulty indices}. Our definition is a simplification of a similar definition in \cite{GZ23}.

\begin{defn}[Faulty index]\label{defn:faulty}
Let $A\in\F_q(X_1,\dots,X_n)^{m\times s}$ be a matrix, where $m\ge s$, and let $A'$ be the $s\times s$ submatrix consisting of the first $s$ rows of $A$. Suppose the entries of $A$ are in $\F_q[X_1,\dots,X_n]$.
For $\alpha_1,\dots,\alpha_n\in \F_q$, we say $i\in [n]$ is the \emph{faulty index} of $A$ (with respect to $\alpha_1,\dots,\alpha_n$) if   $\det(A'|_{X_1=\alpha_1,\dots,X_{i-1}=\alpha_{i-1}})\neq 0$ but $\det(A'|_{X_1=\alpha_1,\dots,X_{i}=\alpha_{i}})= 0$. Note that the faulty index of $A$ is unique if it exists. 
\end{defn}
Next, we will present an algorithm that, when provided with two increasing subsequences $I,J\in[n]^{\ell}$ of length $\ell=2k-1+\floor{\varepsilon n}$ that agree on at most $k-1$ coordinates, elements $\alpha_1,\dots,\alpha_n\in\F_q$, and a parameter $r\in\mathbb{N}^+$, attempts to verify whether the matrix $V_{k,\ell,I,J}|_{X_1=\alpha_1,\dots,X_n=\alpha_n}$ has full column rank. If it is unable to confirm this, the algorithm produces either a ``FAIL'' message or identifies a sequence of faulty indicies $(i_1,\dots,i_r)\in [n]^r$. 
The algorithm is given as  Algorithm~\ref{alg:algorithm}.

\begin{algorithm}[htb]	
	\caption{
	$\mathtt{CertifyFullColumnRankness}$ 
	}
	\label{alg:algorithm}
    \DontPrintSemicolon
	\KwIn{$n,k,r\in\N^+$, increasing subsequences $I,J\in[n]^{\ell}$ with $\ell=2k-1+\floor{\varepsilon n}$ that agree on at most $k-1$ coordinates, and $\alpha_1,\dots,\alpha_n\in \F_q$.} 
	\KwOut{``SUCCESS'', ``FAIL'', or a sequence $(i_1,\dots,i_r)\in [n]^r$.}  
	\BlankLine
	$B\gets\emptyset$.
 
	\For{$j=1$ \KwTo $r$}{
  \uIf{$\rank(V_{k,\ell,I,J}^{B})<2k-1$}{Output ``FAIL'' and halt.}
    \uElseIf(\tcp*[f]{$i$ is unique if it exists}){the faculty index $i\in [n]$ of $V_{k,\ell,I,J}^{B}$ exists  
    }{
    $i_j\gets i$ and  $B\gets B\cup \{i\}$. 
	}
    \Else{Output ``SUCCESS'' and halt.}}
	Output $(i_1,\dots,i_r)$.
\end{algorithm}
\begin{lemma}[Behavior of Algorithm \ref{alg:algorithm}]\label{lem:twocases}
Let $\varepsilon\geq 0$. Let $I,J\in[n]^{\ell}$ be two increasing subsequences that agree on at most $k-1$ coordinates, where $\ell=2k-1+\floor{\varepsilon n}$. Let $r$ be a positive integer such that $r\leq \ceil{\frac{\varepsilon n}{2}}$.
Then for all  $\alpha_1,\dots,\alpha_n\in \F_q$, running Algorithm \ref{alg:algorithm} on the  input $I,J,$ $ \alpha_1,\dots,\alpha_n$, and $r$ yields one of the following two possible scenarios: 
\begin{enumerate} 
\item Algorithm~\ref{alg:algorithm} outputs ``SUCCESS''. In this case, $V_{k,\ell,I,J}|_{X_1=\alpha_1,\dots,X_n=\alpha_n}$ has full column rank.
\item Algorithm~\ref{alg:algorithm} outputs a sequence of distinct indices $(i_1,\dots,i_r)\in [n]^r$. For each $j\in [r]$, $i_j$ is the faulty index of $V_{k,\ell,I,J}^{B_j}$, where $B_j:=\{i_1,\dots,i_{j-1}\}$.
\end{enumerate}    
\end{lemma}
\begin{proof}
We first show that the algorithm never outputs ``FAIL''.
If the algorithm reaches the $j$-th round of the loop, where $j\in [r]$, then at the beginning of this round, we have $|B|=j-1\leq r-1$. 
For any $j\in B$, the number of indices $i\in [\ell]$ for which $I_i=j$ or $J_i=j$ is at most two. 
Therefore, $V_{k,\ell,I,J}^{B}$ is obtained from $V_{k,\ell,I,J}$ by removing at most $2|B|$ rows.
It follows that the number of rows that $V_{k,\ell,I,J}^{B}$ contains is at least $\ell-2|B|\geq (2k-1)+\floor{\varepsilon n} - 2(r-1)\geq 2k-1$.
Note that the top $(2k-1)\times(2k-1)$ submatrix of $V_{k,\ell,I,J}^{B}$ equals $V_{k,2k-1,I^P,J^P}$ for some $P\subseteq [\ell]$ of size $2k-1$.
So by \cref{roub}, the matrix $V_{k,\ell,I,J}^{B}$ has full column rank, which implies that the algorithm never outputs ``FAIL''. 

Assume the algorithm outputs ``SUCCESS'' and halts in the $j$-th round for some $j\in [r]$, which means the faulty index of $V_{k,\ell,I,J}^{B}$
does not exist in that round. 
Let $M$ be the top $(2k-1)\times (2k-1)$ submatrix of $V_{k,\ell,I,J}^{B}$.
Then $\det(M|_{X_1=\alpha_1,\dots,X_{n}=\alpha_{n}})\neq 0$, which implies that $V_{k,\ell,I,J}|_{X_1=\alpha_1,\dots,X_{n}=\alpha_{n}}$ has full column rank.

On the other hand, suppose that the algorithm does not output ``SUCCESS''. Then it outputs some $(i_1,\dots,i_r)\in [n]^r$ where $i_j$ is the faulty index of $V_{k,\ell,I,J}^{B_j}$ and $B_j=\{i_1,\dots,i_{j-1}\}$ for $j\in [r]$. Note that $i_1,\dots,i_r\in [n]$ must be distinct. Indeed, if an index $i$ is in $B_j$, then $X_i$ does not appear in $V_{k,\ell,I,J}^{B_j}$, and hence $i$ cannot be the faulty index of $V_{k,\ell,I,J}^{B_j}$. 
\end{proof}

The next lemma bounds the probability that Algorithm \ref{alg:algorithm} outputs a particular sequence of faulty indices over random $(\alpha_1,\dots,\alpha_n)$. The proof follows the same approach as \cite[Lemma~4.5]{GZ23}.

\begin{lemma}\label{CGLZ24}\label{lem:certprob}
Under the notation and conditions in Lemma~\ref{lem:twocases}, suppose $q\geq n$ and $(\alpha_1,\dots,\alpha_n)$ is chosen uniformly at random from the set of all $n$-tuples of distinct elements in $\F_q$. Then for any sequence $(i_1,\dots,i_r)\in [n]^r$, the probability that Algorithm~\ref{alg:algorithm} outputs $(i_1,\dots,i_r)$ on the input $I,J$, $\alpha_1,\dots,\alpha_n$, and $r$ is at most $\left(\frac{2(k-1)}{q-n+1}\right)^r$.
\end{lemma}
\begin{proof}
For $j\in [r]$, define the following:
\begin{enumerate}
\item $B_j:=\{i_1,\dots,i_{j-1}\}$.
\item Let $A_j$ be the top 
$(2k-1)\times (2k-1)$ submatrix of $V_{k,\ell, I,J}^{B_j}$. (The fact that $\det A_j\neq 0$ symbolically is guaranteed by Corollary \ref{roub}.)
\item Let $E_j$ be the event that $\det(A_j|_{X_1=\alpha_1,\dots,X_{i_j-1}=\alpha_{i_j-1}})\neq 0$ but $\det(A_j|_{X_1=\alpha_1,\dots,X_{i_j}=\alpha_{i_j}})=0$. 
\end{enumerate}
If Algorithm~\ref{alg:algorithm} outputs $(i_1,\dots,i_r)$, then $i_j$ is the faulty index of $V_{k,\ell, I,J}^{B_j}$ for $j\in [r]$ by Lemma~\ref{lem:twocases} and therefore the events $E_1,\dots,E_r$ all occur. So it suffies to verify that $\Pr[E_1\wedge\dots \wedge E_r]\leq \left(\frac{2(k-1)}{q-n+1}\right)^r$.

Let $(j_1, j_2,\dots,j_r)$ be a permutation of $(1,2,\dots,r)$ such that $i_{j_1}<\dots<i_{j_r}$, i.e., $i_{j_\ell}$ is the $\ell$-th smallest index among $i_1,\dots,i_r$ for $\ell\in [r]$. 
For $\ell\in\{0,1,\dots,r\}$, we define $F_\ell:=E_{j_1}\wedge \dots \wedge E_{j_\ell}$, where $F_0$ is the event that always occurs. 
Then $F_r=E_{j_1}\wedge \dots \wedge E_{j_r}=E_1\wedge\dots \wedge E_r$. We may assume $\Pr[F_r]>0$ since otherwise we are done. By definition, if $F_{\ell}$ occurs and $\ell'<\ell$, then $F_{\ell'}$ also occurs. So $\Pr[F_\ell]>0$ for all $\ell\in\{0,1,\dots,r\}$.
We have that
\[
\Pr[E_1\wedge\dots \wedge E_r]=\Pr[F_r]=\prod_{\ell=1}^r \frac{\Pr[F_\ell]}{\Pr[F_{\ell-1}]}\;.
\]
Thus, to prove $\Pr[E_1\wedge\dots \wedge E_r]\leq \left(\frac{2(k-1)}{q-n+1}\right)^r$, it suffices to show that $\frac{\Pr[F_\ell]}{\Pr[F_{\ell-1}]}\leq \frac{2(k-1)}{q-n+1}$ for every $\ell\in[r]$.

Fix $\ell\in [r]$ and let $j=j_\ell$. 
Let  $S$ be the set of all $\beta=(\beta_1,\dots,\beta_{i_{j}-1})\in \F_q^{i_{j}-1}$ such that 
$\Pr\left[\left(\alpha_{<i_{j}}=\beta\right)\wedge F_{\ell-1}\right]>0$,
where $\alpha_{<i_{j}}=\beta$ is a shorthand for  $(\alpha_1=\beta_1)\wedge\dots\wedge (\alpha_{i_{j}-1}=\beta_{i_{j}-1})$.
Note that for $\beta\in S$, the event $\left(\alpha_{<i_{j}}=\beta\right)\wedge F_{\ell-1}$
is simply $\alpha_{<i_{j}}=\beta$ 
since $F_{\ell-1}=E_{j_1}\wedge \dots \wedge E_{j_{\ell-1}}$ depends only on $\alpha_1,\dots,\alpha_{i_{j_{\ell-1}}}$ 
and as $\alpha_{<i_{j}}=\beta\in S$, $F_{\ell-1}$ occurs (by the definition of $S$). 
We then have
\begin{align*}
\frac{\Pr[F_\ell]}{\Pr[F_{\ell-1}]}&=\frac{\sum_{\beta\in S} \Pr\left[\left(\alpha_{<i_{j}}=\beta\right)\wedge F_{\ell}\right]}{\sum_{\beta\in S} \Pr\left[\left(\alpha_{<i_{j}}=\beta\right)\wedge F_{\ell-1}\right]}
=\frac{\sum_{\beta\in S} \Pr\left[\left(\alpha_{<i_{j}}=\beta\right)\wedge  E_{j}\right]}{\sum_{\beta\in S} \Pr\left[ \alpha_{<i_{j}}=\beta \right]}\\
&\leq \max_{\beta\in S}  \frac{ \Pr\left[\left(\alpha_{<i_{j}}=\beta\right)\wedge  E_{j}\right]}{ \Pr\left[ \alpha_{<i_{j}}=\beta \right]}
=\max_{\beta\in S} \Pr\left[  E_{j}\mid \alpha_{<i_{j}}=\beta \right].
\end{align*}
Fix $\beta=(\beta_1,\dots,\beta_{i_{j}-1})\in S$. We just need to prove that $\Pr\left[  E_{j}\mid \alpha_{<i_{j}}=\beta \right]\leq \frac{2(k-1)}{q-n+1}$.
Let \[Q:=\det(A_{j}|_{X_1=\beta_1,\dots,X_{i_j-1}=\beta_{i_{j}-1}})\in\F_q[X_{i_{j}},\dots,X_n].\]
If $Q=0$, then $ E_{j}$ never occurs conditioned on $\alpha_{<i_{j}}=\beta$ and we are done.
So assume  $Q\neq 0$.
View $Q$ as a polynomial in $X_{i_{j}+1},\dots,X_n$ over the ring $\F_q[X_{i_{j}}]$, and let $Q_0\in \F_q[X_{i_{j}}]$ be the coefficient of a nonzero term of $Q$. Then conditioned on $\alpha_{<i_{j}}=\beta$, the event $ E_{j}$ occurs only if $\alpha_{i_{j}}$ is a root of $Q_0\neq 0$.  
Note that $\deg Q_0\leq \deg_{X_{i_{j}}} Q\leq  \deg_{X_{i_{j}}}\left(\det(A_j)\right)\leq 2(k-1)$ by Lemma~\ref{roub}.
Also note that conditioned on $\alpha_{<i_{j}}=\beta$, the random variable $\alpha_{i_{j}}$ is uniformly distributed over the set $\F_q\setminus\{\beta_1,\dots,\beta_{i_{j}-1}\}$, whose size is at least $q-n+1$ since $i_j\leq n$. It follows that $\Pr\left[  E_{j}\mid \alpha_{<i_{j}}=\beta \right]\leq \frac{2(k-1)}{q-n+1}$, as desired.
\end{proof}
By combining Lemma~\ref{lem:twocases} and Lemma~\ref{CGLZ24}  and taking a union bound over the set of possible outputs  $(i_1,\dots,i_r)$ of \cref{alg:algorithm}, we obtain the following corollary.
\begin{cor}\label{cor:fullrankprob}
Under the notation and conditions in Lemma~\ref{lem:twocases}, suppose $q\geq n$ and $(\alpha_1,\dots,\alpha_n)$ is chosen uniformly at random from the set of all $n$-tuples of distinct elements in $\F_q$. 
Then for any positive integer $r\le\ceil{\frac{\varepsilon n}{2}}$, we have 
\[
\Pr\left[\begin{aligned}
    &\text{The matrix }V_{k,\ell,I,J}|_{X_1=\alpha_1,\dots,X_{n}=\alpha_{n}}\\
&\text{does not have full column rank}
\end{aligned}\right]\leq 
\left(\frac{2n(k-1)}{q-n+1}\right)^r.
\]
\end{cor}
\begin{proof}
By Lemma~\ref{CGLZ24} and the union bound, the probability that Algorithm~\ref{alg:algorithm} outputs some sequence $(i_1,\dots,i_r)\in [n]^r$ on the input $I,J$, $\alpha_1,\dots,\alpha_n$, and $r$ is at most 
$n^r\cdot\left(\frac{2(k-1)}{q-n+1}\right)^r= \left(\frac{2n(k-1)}{q-n+1}\right)^r.$ By Lemma~\ref{lem:twocases}, whenever this does not occur, the matrix $V_{k,\ell,I,J}|_{X_1=\alpha_1,\dots,X_n=\alpha_n}$ has full column rank, as desired.
\end{proof}
\subsection{Putting it together}
We are now ready to prove a weaker version of our main result, which achieves quadratic-sized alphabets.
\begin{thm}\label{m1}
Let $\varepsilon\in (0,1)$ and $n,k\in\N^+$, where $k\leq n$.  
Let $q$ be a prime power such that $q\geq \left(1+2\cdot 2^{6/\epsilon}k\right)n$. 
Suppose $(\alpha_1,\dots,\alpha_n)$ is chosen uniformly at random from the set of all $n$-tuples of distinct elements in $\F_q$.
Then with probability at least $1-2^{-n}>0$, the code $\text{RS}_{n,k}(\alpha_1, \ldots, \alpha_n)$ over $\Fq$ corrects at least $(1-\epsilon)n-2k+1$ adversarial insdel errors.
\end{thm}
\begin{proof}
Let $\ell=(2k-1) + \floor{\varepsilon n}$. 
By Lemma \ref{bad}, if $\mathsf{RS}_{n,k}(\alpha_1,\dots,\alpha_n)\subseteq\F_q^n$ fails to correct $n-\ell$  adversarial insdel errors, then we will have two increasing subsequences $I,J\in[n]^{\ell}$ which agree on at most $k-1$ coordinates, such that the matrix $V_{k,\ell,I, J}|_{X_{1}=\alpha_{1},\dots,X_{n}=\alpha_{n}}$ does not have full column rank. On the other hand, applying Corollary \ref{cor:fullrankprob} with $r=\left\lceil\frac{\varepsilon n}{2}\right\rceil$, for fixed $I,J$, we have that
\[\Pr\left[\begin{aligned}
    &\text{The matrix }V_{k,\ell,I,J}|_{X_1=\alpha_1,\dots,X_{n}=\alpha_{n}}\\
&\text{does not have full column rank}
\end{aligned}\right]\leq 
\left(\frac{2n(k-1)}{q-n+1}\right)^{\frac{\epsilon n}{2}}.\]
By the union bound over all $(I,J)$ and the fact that $q\ge \left(1+2\cdot 2^{6/\epsilon}k\right)n$, 
the probability that $\mathsf{RS}_{n,k}(\alpha_1,\dots,\alpha_n)\subseteq\F_q^n$ cannot recover from $(1-\epsilon)n-2k+1$ adversarial insdel errors is at most $\binom{n}{\ell}^2\cdot \left(\frac{2n(k-1)}{q-n+1}\right)^{\frac{\epsilon n}{2}}\leq 2^{2n}\cdot\left(\frac{2n(k-1)}{q-n+1}\right)^{\frac{\epsilon n}{2}}\leq2^{-n}$, as desired. 
\end{proof}

\section{Achieving Linear Alphabet Size}

In this section, we present an improved analysis that achieves a linear alphabet size. We begin by defining the notion of ``chains'' and establishing related structural results about common subsequences in \cref{sec:chain}. Building on these foundations, we then prove \cref{main2} in \cref{sec:main}.

\subsection{Chain decomposition}\label{sec:chain}

\tikzstyle{chainnode}=[circle,draw=none,fill=black,scale=0.5]
\tikzstyle{chain}=[line width=3]
\tikzstyle{chain2}=[dotted,line width=1]
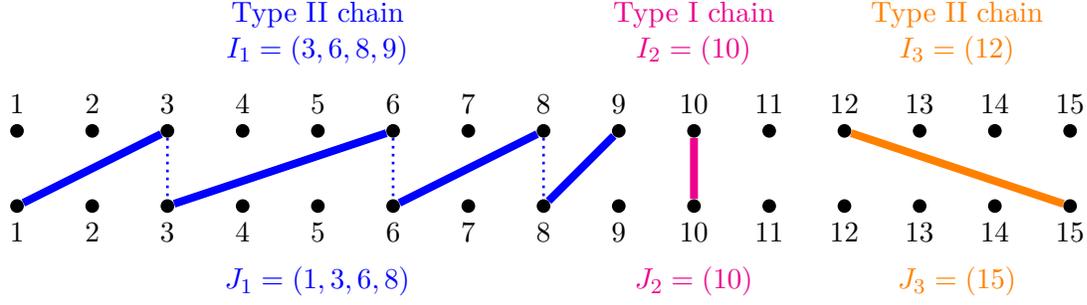
\begin{figure}
  \begin{tikzpicture}[every text node part/.style={align=center}]
        \foreach \i in {1,...,15} {
        \node[chainnode,label=below:{$\i$}] (j\i) at (1*\i,0) {};
        \node[chainnode,label={$\i$}] (i\i) at (1*\i,1) {};
        }

        \draw[chain, color=blue] (j1)--(i3);
        \draw[chain, color=blue] (j3)--(i6);
        \draw[chain, color=blue] (j6)--(i8);
        \draw[chain, color=blue] (j8)--(i9);
        \draw[chain2, color=blue] (i3)--(j3);
        \draw[chain2, color=blue] (i6)--(j6);
        \draw[chain2, color=blue] (i8)--(j8);

        \draw[chain, color=magenta] (j10)--(i10);

        \draw[chain, color=orange] (j15)--(i12);

        \node[color=blue] () at (5,2.3) {Type II chain\\$I_1=(3,6,8,9)$};
        \node[color=blue] () at (5,-1) {$J_1=(1,3,6,8)$};
        
        \node[color=magenta] () at (10,2.3) {Type I chain\\$I_2=(10)$};
        \node[color=magenta] () at (10,-1) {$J_2=(10)$};

        \node[color=orange] () at (13.5,2.3) {Type II chain \\ $I_3=(12)$};
        \node[color=orange] () at (13.5,-1) {$J_3=(15)$};
    \end{tikzpicture}
   
    \caption{Example of chains}\label{exa} 
\end{figure}

For convenience, we first introduce some definitions. 

\begin{defn}
For an increasing subsequence $I=(I_1,\dots,I_{\ell})\in [n]^{\ell}$, define
\[
\range(I):=\{I_1,\dots, I_{\ell}\}.
\]
\end{defn}

\begin{defn}[Chain]
\label{defn:chain}
For two increasing subsequences $I,J\in[n]^\ell$, where $n,\ell\in\N^+$, we call the pair $(I,J)$ a \textit{chain} if either (i) $I_{i} = J_{i+1}$ for all $i=1,\dots,\ell-1$, or (ii) $I_{i+1} = J_{i}$ for all $i=1,\dots,\ell-1$.
If $\ell=1$ and $I=J$, we call $(I,J)$ a Type I chain. Otherwise, we call $(I,J)$ a Type II chain. See Figure \ref{exa} for some examples. 
\end{defn}

Recall \cref{defn:restriction} that for an increasing subsequence $I\in [n]^{\ell}$ and $P=\{i_1,\dots,i_{\ell'}\}\subseteq [\ell]$ with $i_1<\dots<i_{\ell'}$, we let $I^P=(I_{i_1},\dots,I_{i_{\ell'}})$.

\begin{defn}[Maximal chain]\label{defn:maximal-chain}
Let $I,J\in[n]^{\ell}$ be two increasing subsequences.
Let $P$ be a subset of $[\ell]$ of size $\ell'\geq 1$.
Let $I'=I^P$ and $J'=J^P$. Suppose $(I',J')$ is a chain. 
We say the chain $(I',J')$ is \emph{maximal} with respect to $(I,J)$ if 
\begin{enumerate}
\item $(I',J')$ is a Type I chain, or 
\item $(I',J')$ is a Type II chain, and both $\min(I'_1,J'_1)$ and $\max(I'_{\ell'},J'_{\ell'})$ is in exactly one of $\range(I)$ and $\range(J)$.
\end{enumerate}
\end{defn}

Intuitively, a chain $(I', J')$ is maximal with respect to $(I,J)$ if it cannot be extended to a longer chain in $(I,J)$ by adding matches from either side.

\begin{defn}[Indices of variables]
For increasing subsequences $I,J\in [n]^{\ell}$ and $P\subseteq [\ell]$, define
\[
\var(I,J,P):=\range(I^P)\cup\range(J^P)\subseteq [n].
\]
Note that $\var(I,J,P)$ is the set of indices of the variables involved in a $V$-matrix $V_{k,|P|,I^P, J^P}$.\footnote{If $k=1$, these variables do not really appear as $X_i^{k-1}=X_i^0=1$. But we still consider $\var(I,J,P)$ as the set of indices of variables involved in $V_{k,|P|,I^P, J^P}$ in this case.}
\end{defn}

\begin{defn}[$(I,J)$-disjointness]\label{defn:disjointness}
Let $I,J\in[n]^\ell$ be increasing subsequences.
We say two sets $P,P'\subseteq [\ell]$ are \emph{$(I,J)$-disjoint} if $\var(I,J,P)$ and $\var(I,J,P')$ are disjoint. Note that $(I,J)$-disjointness implies disjointness.
\end{defn}

\begin{lemma}\label{lemma:disjointness}
Let $I,J\in[n]^{\ell}$ be two increasing subsequences.
Let $P'\subseteq P\subseteq [\ell]$. 
Suppose $(I^{P'}, J^{P'})$ is a maximal chain with respect to $(I^P,J^P)$.
Then $P'$ and $P\setminus P'$ are $(I,J)$-disjoint.
\end{lemma}
\begin{proof}
By definition, we have
\begin{equation}\label{eq:complement}
\range(I^{P\setminus P'})=\range(I^P)\setminus \range(I^{P'}) \quad\text{and}\quad \range(J^{P\setminus P'})=\range(J^P)\setminus \range(J^{P'}).
\end{equation}
If $(I^{P'},J^{P'})$ is a Type I chain, the claim holds by \eqref{eq:complement} and the fact the that $I^{P'}=J^{P'}$.

Now suppose $(I^{P'},J^{P'})$ is a Type II chain.
Let $\ell'=|P'|$.
By definition, either $(I^{P'})_{i} = (J^{P'})_{i+1}$ for $i=1,\dots,\ell'-1$, or $(I^{P'})_{i+1} = (J^{P'})_{i}$ for $i=1,\dots,\ell'-1$.
This implies that all elements in $\var(I,J,P')$ except $a:=\min((I^{P'})_1,(J^{P'})_1)$ and $b:=\max((I^{P'})_{\ell'},(J^{P'})_{\ell'})$ are in both $\range(I^{P'})$ and $\range(J^{P'})$. So by \eqref{eq:complement}, these elements are excluded from $\var(I,J,P\setminus P')$.
As $(I^{P'}, J^{P'})$ is maximal with respect to $(I^P,J^P)$, both $a$ and $b$ are in exactly one of $\range(I^P)$ and $\range(J^P)$. 
And they are in $\range(I^P)$ (resp. $\range(J^P)$) iff they are in $\range(I^{P'})$ (resp. $\range(J^{P'})$).
By \eqref{eq:complement}, $a$ and $b$ are excluded from $\var(I,J,P\setminus P')$ as well.
\end{proof}
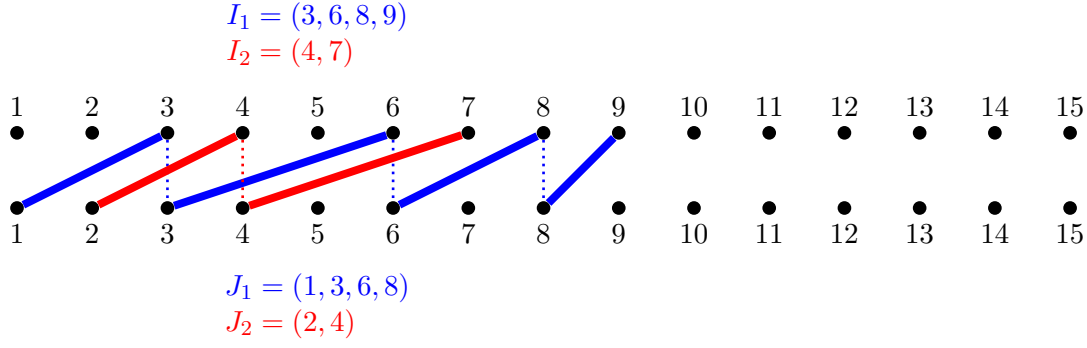
\begin{figure}
      \begin{tikzpicture}[every text node part/.style={align=left}]
        \foreach \i in {1,...,15} {
        \node[chainnode,label=below:{$\i$}] (j\i) at (1*\i,0) {};
        \node[chainnode,label={$\i$}] (i\i) at (1*\i,1) {};
        }
        \draw[chain, color=blue] (j1)--(i3);
        \draw[chain, color=blue] (j3)--(i6);
        \draw[chain, color=blue] (j6)--(i8);
        \draw[chain, color=blue] (j8)--(i9);
        \draw[chain2, color=blue] (i3)--(j3);
        \draw[chain2, color=blue] (i6)--(j6);
        \draw[chain2, color=blue] (i8)--(j8);

        \draw[chain, color=red] (j2)--(i4);
        \draw[chain, color=red] (j4)--(i7);
        \draw[chain2, color=red] (i4)--(j4);
        
        \node[] () at (5,2.3) {{\color{blue}$I_1=(3,6,8,9)$}\\{\color{red}$I_2=(4,7)$}};
        \node[] () at (5,-1.3) {{\color{blue}$J_1=(1,3,6,8)$}\\{\color{red}$J_2=(2,4)$}};
    \end{tikzpicture}
    \caption{A decomposition of $(I,J)$, where $I=(3,4,6,7,8,9)$ and $J=(1,2,3,4,6,8)$, into maximal chains $(I_1,J_1)$ and $(I_2,J_2)$. Note that the maximal chains are $(I,J)$-disjoint but can interleave.} 
    \label{fig:decomposition}
\end{figure}

\begin{thm}\label{thm:structure-general}
Let $I,J\in [n]^\ell$ be increasing subsequences.
Let $P\subseteq [\ell]$.
Then there exists a partition $P_1\sqcup P_2\sqcup \dots\sqcup P_s$ of $P$ into nonempty sets $P_i$ such that $(I^{P_i},J^{P_i})$ is a maximal chain with respect to $(I^P,J^P)$ for all $i\in[s]$. 
\end{thm}

\begin{proof}
Induct on $t:=|P|$. When $t=0$, i.e., $P=\emptyset$, the theorem holds trivially. Now assume $t>0$ and the theorem holds for $t'<t$.

We first construct a nonempty set $P_1\subseteq P$ such that $(I^{P_1},J^{P_1})$ is a maximal chain with respect to $(I^P,J^P)$. Initially, let $i_1$ be the smallest integer in $P$, and let $P_1=\{i_1\}$.
If $I_{i_1}=J_{i_1}$, then $(I^{P_1},J^{P_1})$ is a Type I chain and hence a maximal chain with respect to $(I^P,J^P)$ by \cref{defn:maximal-chain}.

Now assume $I_{i_1}<J_{i_1}$. In this case, run the following process on $P_1=\{i_1\}$:  
{
\setlength{\interspacetitleruled}{0pt}%
\setlength{\algotitleheightrule}{0pt}%
\begin{algorithm}[htb]	
$s\gets 1$.

\While(\tcp*[f]{$i$ is unique if it exists}){there exists $i\in P$ such that $I_{i}=J_{i_s}$}{
$s\gets s+1$.

$i_s\gets i$.

$P_1\gets P_1\cup\{i_s\}$.
}
\end{algorithm}
}

The above process iteratively finds indices $i_1<\dots<i_s$ in $P$ such that $I_{i_1}<J_{i_1}=I_{i_2}<J_{i_2}=\dots=I_{i_s}<J_{i_s}$, and $J_{i_s}\not\in\range(I^{P})$. The resulting set $P_1$ is $\{i_1,\dots,i_s\}$. 
Note that both $I_{i_1}=\min(I_{i_1},J_{i_1})$ and $J_{i_s}=\max(I_{i_s}, J_{i_s})$ are in exactly one of $\range(I^{P})$ and $\range(J^{P})$. 
So by \cref{defn:maximal-chain}, $(I^{P_1}, J^{P_1})$ is a maximal chain with respect to $(I^P, J^P)$.

Finally, in the case where $I_{i_1}>J_{i_1}$, we construct $P_1$ via a symmetric process, where $I$ and $J$ and swapped. 

In all cases, we obtain a nonempty set $P_1\subseteq P$ such that $(I^{P_1},J^{P_1})$ is a maximal chain with respect to $(I^P,J^P)$. 
By \cref{lemma:disjointness}, $P_1$ and $P\setminus P_1$ are $(I,J)$-disjoint.

By the induction hypothesis, there exists a partition $P_2\sqcup \dots\sqcup P_s$ of $P\setminus P_1$ such that $(I^{P_i},J^{P_i})$ is a maximal chain with respect to $(I^{P\setminus P_1},J^{P\setminus P_1})$ for $i=2,\dots,s$. 
As $P_1$ and $P\setminus P_1$ are $(I,J)$-disjoint, it follows from \cref{defn:maximal-chain} that a maximal chain with respect to $(I^{P\setminus P_1},J^{P\setminus P_1})$ is also maximal with respect to $(I^P,J^P)$.
So $(I^{P_i},J^{P_i})$ is a maximal chain with respect to $(I^P,J^P)$ for $i=1,2,\dots,s$, as desired. 
\end{proof}

By choosing $P=[\ell]$, we obtain the following corollary, which states that for increasing subsequences $I,J\in [n]^\ell$, $(I,J)$ can always be decomposed into maximal chains.

\begin{cor}\label{thm:structure}
For any two increasing subsequences $I,J\in[n]^{\ell}$, there exists a partition $P_1\sqcup P_2\sqcup \dots\sqcup P_s$ of $[\ell]$ such that $(I^{P_i},J^{P_i})$ is a maximal chain with respect to $(I,J)$ for all $i\in[s]$. 
\end{cor}

See \cref{fig:decomposition} for an example. In fact, one can prove that the decomposition into maximal chains is unique, but we do not need this result.

\paragraph{Splitting long chains into short ones.} 
The maximal chains in \cref{thm:structure} can be very long. The next lemma states that, by removing a small fraction of matchings from $(I,J)$, we can split the long chains into very short ones while maintaining their $(I,J)$-disjointness. (See also \cref{fig:split}.)

\begin{lemma}
\label{lem:split}
Let $I, J\in [n]^\ell$ be increasing subsequences. Let $\epsilon\in (0,1)$.
Then there exist a subset $P\subseteq [\ell]$ of size at least $(1-\eps)\ell$ and a partition 
$P_1\sqcup P_2\sqcup \dots\sqcup P_s$ of $P$ 
such that,
\begin{enumerate}
    \item $P_1,\dots,P_s$ are mutually $(I,J)$-disjoint.
    \item $|P_i|\leq 1/\eps$ for $i\in [s]$.
\end{enumerate}
\end{lemma}
\begin{figure}
\centering
      \begin{tikzpicture}[scale=0.5,every text node part/.style={align=left}]
      
        \foreach \i in {1,...,30} {
            \node[chainnode] (j\i) at (1*\i,0) {};
            \node[chainnode] (i\i) at (1*\i,1) {};
        }       
        \draw[chain, color=blue] (1,0)  -- (2,1);
        \foreach \i in {2,...,29} {
            \draw[chain2, color=blue] (1*\i,0)  -- (1*\i,1);
            \draw[chain, color=blue] (1*\i,0)  -- (1*\i+1,1);
        }
        \node[] at (15.5,-1) {$\downarrow$};
        \foreach \i in {1,...,30} {
            \node[chainnode] (j\i) at (1*\i,-2) {};
            \node[chainnode] (i\i) at (1*\i,-3) {};
        }       
        \draw[chain, color=blue] (1,0)  -- (2,1);
        \foreach \i in {0,...,4} {
            \tikzmath{\x = 100;}
            \draw[chain, color=blue!\x] (6*\i+1,-3)  -- (6*\i+2,-2);    
            \foreach \j in {2,...,5} {
                \draw[chain2, color=blue!\x] (6*\i+\j,-3)  -- (6*\i+\j,-2);
                \draw[chain, color=blue!\x] (6*\i+\j,-3)  -- (6*\i+\j+1,-2);
            }
            \node[rotate=90] at (3.5+6*\i,-3.6) {\Huge \scalebox{1}[3]{$\{$}};
            \node[] at (3.5+6*\i,-4.5) {$1/\varepsilon$};
        }

    \end{tikzpicture}
    \caption{\cref{lem:split}: splitting long chains into short ones. Ensure that all chains have a length of at most $1/\eps$ by removing at most $\eps$ fraction of pairs from each chain.}
    \label{fig:split}
\end{figure}
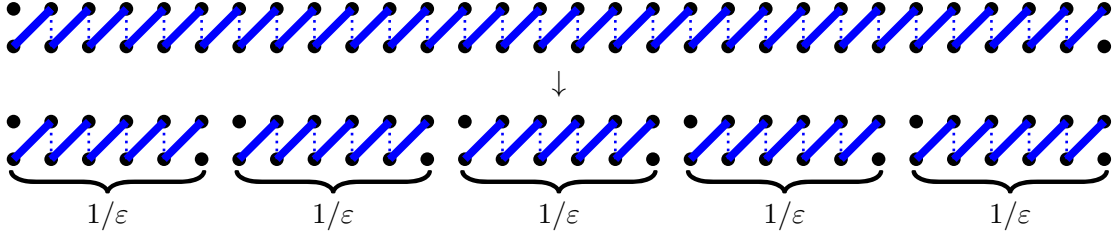

\begin{proof}
By \cref{thm:structure}, there exists a partition $P_1\sqcup P_2\sqcup \dots\sqcup P_{s}$ of $[\ell]$ such that $(I^{P_i},J^{P_i})$ is a maximal chain with respect to $(I,J)$ for all $i\in[s]$. Fix such a partition.
Then Item~1 holds by \cref{lemma:disjointness}.
Consider any $i\in [s]$ such that $|P_i|>1/\eps$.
Note that $(I^{P_i},J^{P_i})$ must be a Type II chain (since otherwise we would have $|P_i|=1\leq 1/\eps$).
Let $t=\lfloor 1/\eps\rfloor+1\geq 1/\eps$.
We remove the $j$-th smallest element of $P_i$ for $j=t, 2t,\dots,\lfloor P_i/t\rfloor t$.
The remaining elements split into consecutive blocks $P_{i,j}$, each of size at most $t-1\leq 1/\eps$. 
Splitting $P_i$ into $P_{i,j}$ preserves the $(I,J)$-disjointness of the sets. See \cref{fig:split} for an illustration. 

Perform the above operations for each $i$ with $|P_i|>1/\eps$, and we obtain the desired partition. The set $P$ consists of the remaining elements of $[\ell]$. Its size is at least $(1-\eps)\ell$ since we have removed at most $(1/t)$-fraction of elements from each $P_i$, where $1/t\leq \eps$.
\end{proof}
\subsection{Proof of Theorem \ref{main2}}\label{sec:main}

Fix $I,J\in [n]^\ell$ to be increasing subsequences that agree on at most $k-1$ coordinates.
For convenience, we introduce the following notation.

\begin{defn}\label{defn:blocks}
Let $n,k,\ell\in\N^+$. Let $P_1,\dots,P_s\subseteq [\ell]$ be nonempty subsets of $[\ell]$.
Define the block matrix
\[
\widetilde{V}_{k,I,J,(P_i)_{i=1}^s}(X_1,\dots,X_n)=\begin{pmatrix}
    V_{k,|P_1|,I^{P_1},J^{P_1}}\\
    V_{k,|P_2|,I^{P_2},J^{P_2}}\\
    \vdots\\
    V_{k,|P_s|,I^{P_s},J^{P_s}}
\end{pmatrix}
\]
which is a $\left(\sum_{i=1}^s |P_i|\right) \times (2k-1)$ matrix over $\F_q(X_1,\dots,X_n)$.
\end{defn}
\begin{lemma}\label{lem:submatrix}
Let $P_1,\dots,P_s\subseteq [\ell]$ and $\alpha_1,\dots,\alpha_n\in\F_q$.
If $\widetilde{V}_{k,I,J,(P_i)_{i=1}^s}$ has full column rank under the assignment $X_1=\alpha_1,\dots,X_n=\alpha_n$, then $V_{k,\ell,I,J}$ also has full column rank under the same assignment.
\end{lemma}
\begin{proof}
By definition, each row of $\widetilde{V}_{k,I,J,(P_i)_{i=1}^s}|_{X_1=\alpha_1,\dots,X_n=\alpha_n}$ is a row of $V_{k,\ell,I,J}|_{X_1=\alpha_1,\dots,X_n=\alpha_n}$.
So if the former matrix has full column rank, i.e., its row space has dimension $2k-1$, then so does the latter matrix. 
\end{proof}

\paragraph{Choosing sets $P_1,\dots,P_s$.}
Fix a parameter $\eps_0\in (0,1)$, whose exact value will be determined later. 
We note that by 
\cref{lem:split}, there exist mutually $(I,J)$-disjoint nonempty sets $P_1,\dots,P_s\subseteq [\ell]$, each of size at most $1/\eps_0$, such that $\sum_{i=1}^s |P_i|\geq (1-\eps_0)\ell$ and each $(I^{P_i}, J^{P_i})$ is a chain.
Fix such $P_1,\dots,P_s$.
We further make the following assumption:
\begin{assumption}\label{order-on-Pi}
$|P_1|\leq |P_2|\leq \dots \leq |P_s|$. Moreover, for $i,j\in [s]$ with $i<j$, if $(I^{P_j}, J^{P_j})$ is a Type I chain, so is $(I^{P_i}, J^{P_i})$. In other words, the Type I chains $(I^{P_i}, J^{P_i})$ have the smallest indices $i$.
\end{assumption}
Note that \cref{order-on-Pi} can be guaranteed by permuting the sets $P_i$. 

We will show that $\widetilde{V}_{k,I,J,(P_i)_{i=1}^s}|_{X_1=\alpha_1,\dots,X_n=\alpha_n}$ has full column rank with high probability over random $(\alpha_1,\dots,\alpha_n)$. By \cref{lem:submatrix}, this would imply that $V_{k,\ell,I,J}|_{X_1=\alpha_1,\dots,X_n=\alpha_n}$ also has full column rank with high probability.

\begin{algorithm}[htbp]
\caption{$\mathtt{CertifyFullColumnRankness2}$}
	\label{alg:algorithm-linear}
    \DontPrintSemicolon
	\KwIn{$n,k,\ell,r\in\N^+$, $\eps_0\in (0,1)$, increasing subsequences $I,J\in[n]^{\ell}$, $P_1,\dots,P_s\subseteq [\ell]$, and $\alpha_1,\dots,\alpha_n\in \F_q$.}  
	\KwOut{``SUCCESS'' or a sequence $(i_1,\dots,i_r)\in \{0,\dots,2k-2\}^r$.}  
	\BlankLine
    $S\gets \emptyset$. \tcp*[f]{set of indices of assigned variables} \label{line:1}
    
	$c\gets 0$. \tcp*[f]{number of assigned rows} \label{line:2}
    
    Pick the smallest $m\in [s]$ such that $\sum_{i=s-m+1}^s |P_i|\in [r/\eps_0, (r+1)/\eps_0]$. \tcp*[f]{last $m$ chains form the bank} \label{line:3}

    $a\gets 1$. \tcp*[f]{current chain to assign} \label{line:4}

    $b\gets s-m+1$. \tcp*[f]{current chain in the bank} \label{line:5}

    $j\gets 0$. \tcp*[f]{number of failed attempts} \label{line:6}

    \lFor{$i=1$ to $s$}{$Q_i\gets P_i$.} \label{line:7}
 
    \While{$c<2k-1$\label{line:8}} 
    {
    $M\gets$ the top $(2k-1)\times (2k-1)$ submatrix of $\widetilde{V}_{k,I,J,(Q_i)_{i=1}^{s-m}}$. \label{line:9}

    $S'\gets S\cup\var(I,J,Q_a)$. \label{line:10}

    $\overline{M}\gets M|_{X_i=\alpha_i \text{ for } i\in S'}$. \label{line:11}

    \uIf(\tcp*[f]{assign variables in the $a$-th chain}){$\overline{M}$ is nonsingular\label{line:12}}
        {
        $S\gets S'$. \label{line:13}

        $c\gets c + |Q_a|$. \label{line:14}

        $a\gets a+1$. \label{line:15}
        }
    \Else(\tcp*[f]{replace the $a$-th chain by (part of) the $b$-th chain in the bank})
        {

        $Q_a\gets$ set of the smallest $|Q_a|$ elements in $Q_{b}$.
        \label{line:17} 

        $Q_b\gets \emptyset$. \label{line:18}
        
        $b\gets b+1$. \label{line:increase-b}

        $j\gets j+1$. \label{line:22}

        $i_j\gets c$. \label{line:23}
        
        \lIf{$j=r$}{output $(i_1,\dots,i_r)$ and halt. \label{line:24}}
        }
    }
    Output ``SUCCESS''. \label{line:27}
\end{algorithm}
To bound the probability that $\widetilde{V}_{k,I,J,(P_i)_{i=1}^s}|_{X_1=\alpha_1,\dots,X_n=\alpha_n}$ does not have full column rank, we use \cref{alg:algorithm-linear}, whose pseudocode is given above, to assign the variables chain by chain. We sketch how the algorithm works.
\begin{itemize}
\item First, choose the smallest $m\in [s]$ such that $\sum_{i=s-m+1}^s |P_i|\in [r/\eps_0, (r+1)/\eps_0]$ (\cref{line:3}). We will soon prove that such $m$ exists for $\ell$ slightly larger than $2k-1$. The last $m$ chains $(I^{P_i}, J^{P_i})$, $i=s-m+1,\dots,s$, form a ``bank.'' At a high level, \cref{alg:algorithm-linear} attempts to assign the variables in the first $s-m$ chains, i.e., those not in the bank, while keeping the top $(2k-1)\times (2k-1)$ submatrix of $\widetilde{V}_{k,I,J,(P_i)_{i=1}^{s-m}}$ nonsingular even after the assignment.
If assigning the variables in some chain violates this property, the algorithm will use part of a chain in the bank to replace that chain, and continue. Such a replacement is done by updating the sets $P_i$. To avoid confusion, we create a copy $Q_i=P_i$ for $i\in [s]$ (\cref{line:7}) and update the sets $Q_i$ instead.

\item The variables are assigned chain by chain. That is, for $a=1,2,\dots$, \cref{alg:algorithm-linear} assigns $\alpha_i$ to $X_i$ for all $i\in \var(I,J,Q_a)$.
The algorithm terminates when the top $2k-1$ rows of $\widetilde{V}_{k,I,J,(Q_i)_{i=1}^{s-m}}$ have been fully assigned. This condition is checked at \cref{line:8}. The number of assigned rows is maintained as a variable $c$. That is, if the first $i$ chains $(I^{Q_1}, J^{Q_1}),\dots, (I^{Q_i}, J^{Q_i})$ have been assigned, then $c=\sum_{j=1}^{i} |Q_j|$.

\item 
When attempting to assign the variables in the $a$-th chain $(I^{Q_a}, J^{Q_a})$, we check (\cref{line:12}) whether the top $(2k-1)\times (2k-1)$ submatrix of $\widetilde{V}_{k,I,J,(Q_i)_{i=1}^{s-m}}$ remains nonsingular after the assignment. If so, we perform the assignment, and then move to the next chain (\crefrange{line:13}{line:15}). If not, we use a prefix of a chain $(I^{Q_{b}}, J^{Q_{b}})$ in the bank to replace the chain $(I^{Q_a}, J^{Q_a})$, and then throw away $(I^{Q_{b}}, J^{Q_{b}})$.\footnote{Alternatively, one can remove the used portion of $(I^{Q_{b}}, J^{Q_{b}})$ and continue using the remaining part. However, this approach does not qualitatively improve our parameters.}
In the algorithm, the updates to the chains $(I^{Q_a}, J^{Q_a})$ and $(I^{Q_{b}}, J^{Q_{b}})$ are done by simply updating the sets $Q_a$ and $Q_{b}$ (\crefrange{line:17}{line:18}). 

\item Finally, the algorithm maintains a variable $j$, which is the number of times the nonsingularity condition at \cref{line:12} fails to hold.
It also maintains a sequence $(i_1,\dots,i_j)$. 
When the nonsingularity condition fails for the $j$-th time, the number of successfully assigned rows is stored into $i_j$ (\cref{line:23}). We have $i_j\leq 2k-2$ due to the condition checked at \cref{line:8}.
If $j$ reaches a given integer $r$, the algorithm outputs the sequence $(i_1,\dots,i_r)$ and halts (\cref{line:24}).
We will eventually choose $r=\Theta(\eps^2 n)$.
\end{itemize}
From now on, assume $\ell$ is chosen such that
\begin{equation}\label{eq:condition}
(1-\eps_0)\ell\geq 2k-1+(r+1)/\eps_0. 
\end{equation}

As promised, we now show that \cref{line:3} of \cref{alg:algorithm-linear} is valid by proving the following lemma.
\begin{lemma}\label{lem:sum-of-Pi}
There exists $m\in [s]$ such that $\sum_{i=s-m+1}^s |P_i|\in [r/\eps_0, (r+1)/\eps_0]$. Moreover, for such $m$, it holds that $\sum_{i=1}^{s-m}|P_i|\geq 2k-1$ and $m\geq r$. 
\end{lemma}
\begin{proof}
Recall that $P_1,\dots,P_s$ that we choose satisfy
$\sum_{i=1}^s |P_i|\geq (1-\eps_0)\ell$.
Combining this with \eqref{eq:condition}, we have that $\sum_{i=1}^s |P_i|\geq (1-\eps_0)\ell\geq r/\eps_0$. Thus, there
exists $m\in [s]$ such that $\sum_{i=s-m+1}^s |P_i|\geq r/\eps_0$.
Pick the smallest such $m$. Then $\sum_{i=s-m+1}^s |P_i|\leq (r+1)/\eps_0$ by the minimality of $m$ and the fact that the size of each $P_i$ is at most $1/\eps_0$.
Therefore,
\[
\sum_{i=1}^{s-m}|P_i|=\left(\sum_{i=1}^{s}|P_i|\right)-\left(\sum_{i=s-m+1}^{s}|P_i|\right)\geq (1-\eps_0)\ell - (r+1)/\eps_0 \stackrel{\eqref{eq:condition}}{\geq} 2k-1. 
\]
Finally, as $\sum_{i=s-m+1}^s |P_i|\geq r/\eps_0$ and each $P_i$ has size at most $1/\eps_0$, we have $m\geq r$. 
\end{proof}

The next lemma ensures the validity of the other parts of \cref{alg:algorithm-linear}. Particularly, \cref{item:b} guarantees that \cref{line:9} is valid, while \cref{item:c,item:d} guarantees that \crefrange{line:17}{line:18} are valid.

\begin{lemma}\label{lem:basic-claims}
During \cref{alg:algorithm-linear}, after Line~7:
\begin{enumerate}[(a)]
\item\label{item:a}
 $|Q_i|=|P_i|$ for $i\in [s-m]$. In particular, $|Q_i|\leq 1/\eps_0$ for $i\in [s-m]$.
\item\label{item:b}
$\sum_{i=1}^{s-m}|Q_i|\geq 2k-1$.
\item\label{item:c0}
$1\leq a\leq s-m$ always holds at \cref{line:9}.
\item\label{item:c}
$s-m+1\leq b\leq s$ always holds at \cref{line:9}.
\item\label{item:d}
$|Q_{b}|\geq |Q_a|$ always holds at \cref{line:9}, and consequently, also at \cref{line:17}.

\end{enumerate}
\end{lemma}

\begin{proof} 
First, note that the algorithm runs \crefrange{line:17}{line:24} at most $r$ times due to \cref{line:24} and the fact that each time it runs, $b$ increases by one at \cref{line:increase-b}.
Initially, $b=s-m+1$.
As $r\leq m$ by \cref{lem:sum-of-Pi},  \cref{item:c} holds.

Note that $\sum_{j=1}^{s-m}|P_i|\geq 2k-1$ holds by \cref{lem:sum-of-Pi}. Thus, if \ref{item:a} is true, so is \ref{item:b}. 
And if \ref{item:b} is true,  \ref{item:c0} must also be true.
This is because if $a>s-m$ ever holds at \cref{line:9}, then by \ref{item:b}, we have $c=\sum_{i=1}^{a-1} |Q_i|\geq 2k-1$ and the algorithm would have exited the while loop before at \cref{line:8}.

We know \ref{item:a} holds initially after \cref{line:7}. It continues to hold in each iteration of the while loop since \cref{line:17} does not change the size of $Q_a$ and \cref{line:18} only changes $Q_b$ for $b>s-m$.

Finally, to see that \ref{item:d} holds, note that at \cref{line:9}, we have $Q_b=P_b$. This is because $Q_b$ is only changed at \cref{line:18}, followed by  \cref{line:increase-b} which increases $b$ by one. So $Q_b=P_b$ still holds after \cref{line:increase-b}.
Then by \ref{item:a}, \ref{item:c0}, \ref{item:c}, and the fact that $|P_1|\leq |P_2|\leq \cdots\leq |P_s|$, we have $|Q_b|=|P_b|\geq |P_a|=|Q_a|$ at  \cref{line:9}, proving \ref{item:d}.
\end{proof}

The output of \cref{alg:algorithm-linear} can be used to certify the full-column-rankness of the matrix $V_{k,\ell,I,J}|_{X_1=\alpha_1,\dots,X_n=\alpha_n}$, as stated by the following lemma.

\begin{lemma}\label{lem:terminate} 
\cref{alg:algorithm-linear} terminates, outputting either ``SUCCESS'' or $(i_1,\dots,i_r)\in\{0,\dots,2k-2\}^r$. In the former case, $V_{k,\ell,I,J}|_{X_1=\alpha_1,\dots,X_n=\alpha_n}$ has full column rank. 
\end{lemma}

\begin{proof}
In each iteration of the while loop, either $a$ or $j$ increases by one. If $j$ reaches $r$, the algorithm outputs $(i_1,\dots,i_r)$ and halts, where $i_1,\dots,i_r\leq 2k-2$ due to the condition checked at \cref{line:8}.
On the other hand, if $a$ gets large enough, the number of assigned rows $c=\sum_{i=1}^{a-1} |Q_i|$ becomes at least $2k-1$ by \cref{lem:basic-claims}, in which case the algorithm outputs ``SUCCESS'' and terminates. 
Suppose this is the case. Then the condition at \cref{line:12} is satisfied in the last round of the while loop, i.e., the top $(2k-1)\times (2k-1)$ submatrix of $\widetilde{V}_{k,I,J,(Q_i)_{i=1}^{s-m}}|_{X_i=\alpha_i \text{ for } i\in S'}$, denoted by $\overline{M}$, is nonsingular.
As the top $2k-1$ rows have been fully assigned, $\overline{M}$ is also the top $(2k-1)\times (2k-1)$ submatrix of
$\widetilde{V}_{k,I,J,(Q_i)_{i=1}^{s-m}}|_{X_1=\alpha_1,\dots,X_n=\alpha_n}$.
By \cref{lem:submatrix}, $V_{k,\ell,I,J}|_{X_1=\alpha_1,\dots,X_n=\alpha_n}$ has full column rank.
\end{proof}

As the variable $c$ in \cref{alg:algorithm-linear} never decreases, we have the following lemma.

\begin{lemma}\label{lem:output}
Any sequence $(i_1,\dots,i_r)$ output by \cref{alg:algorithm-linear} satisfies $i_1\leq \dots\leq i_r$.
\end{lemma}

Recall that for $i\in [s]$, the set of indices of the variables involved in $V_{k,|Q_i|,I^{Q_i},J^{Q_i}}$ is $\var(I,J,Q_i)$.
The next lemma states that these sets are mutually $(I,J)$-disjoint (see \cref{defn:disjointness}).

\begin{lemma}\label{lem:disjoint}
$Q_1,\dots,Q_s$ are mutually $(I,J)$-disjoint before and after \crefrange{line:17}{line:18}.
\end{lemma}

\begin{proof}
Initially, the $(I,J)$-disjointness holds since it holds for the sets $P_1,\dots,P_s$ and $Q_i=P_i$ for $i\in [s]$.
We just need to verify that \crefrange{line:17}{line:18} preserve the $(I,J)$-disjointness.
At \cref{line:17}, $Q_a$ is replaced by a subset of $Q_b$, and at \cref{line:18}, we let $Q_b=\emptyset$. So the $(I,J)$-disjointness still holds.
\end{proof}

The next lemma states that the
values $\alpha_i$ with $i\in\var(I,J,Q_a)$ used in one iteration of the while loop have not been used in the previous iterations.

\begin{lemma}\label{lem:new-elements}
At \cref{line:9} of \cref{alg:algorithm-linear}, none of the elements in $\var(I,J,Q_a)$ has been added to $S'$ in the previous iterations of the while loop.
\end{lemma}

\begin{proof}
In each iteration of the while loop, either $a$ increases by one, or we replace $Q_a$ by a subset of $Q_b$ and let $Q_b=\emptyset$. Either case, the set $\bigcup_{i=a}^{s} \var(I,J,Q_i)$ shrinks to a proper subset and excludes the elements in $\var(I,J,Q_a)$ (for the old $a$). 

Suppose an element $x\in [n]$ is added to $S'$ for the first time at \cref{line:10}. Then $x\in \var(I,J,Q_a)$. (This is because, if $x\in S$ at \cref{line:10}, then $x$ must have been added to $S'$ before.) 
By the first paragraph, $x$ will not appear in $\var(I,J,Q_a)$ in the later iterations. The contrapositive of this statement is the lemma.
\end{proof}

Next, we make the crucial observation that the behavior of the algorithm is mostly determined by the values of $j$, $(i_1,\dots,i_j)$, and $c$, regardless of $\alpha_1,\dots,\alpha_n$. 

\begin{lemma}\label{lem:determined}
At \cref{line:12} of \cref{alg:algorithm-linear}, the values of $Q_1,\dots,Q_s$ and $a$ are determined by the values of $j$, $i_1,\dots,i_j$, and $c$, together with the input excluding $\alpha_1,\dots,\alpha_n$. 
\end{lemma}

\begin{proof}
Fix the input excluding $\alpha_1,\dots,\alpha_n$ and the values of $j$, $i_1,\dots,i_j$, and $c$.
We use $j^*$, $i_1^*,\dots, i_{j^*}^*$, and $c^*$ to denote these values and distinguish them from the variables.

Simulate the algorithm until it reaches \cref{line:12} with $j=j^*$, $(i_1,\dots,i_j)=(i_1^*,\dots, i_j^*)$, and $c=c^*$. Then read off $Q_1,\dots, Q_s$ and $a$. Doing the simulation naively requires knowing $\alpha_1,\dots,\alpha_n$ to check the condition at \cref{line:12}.
However, we will show that the outcome of any condition check at \cref{line:12} performed during the simulation can be determined by our fixed values.

Suppose the simulation has reached \cref{line:12}, where $i_1,\dots,i_{j}$ have been chosen and $j\leq j^*$. 
First consider the case $j<j^*$.
In this case, as $j$ has not yet reached $j^*$, the values $i_{j+1}^*, \ldots,i_{j^*}^*$ have not yet been assigned to $i_{j+1},\ldots, i_{j^*}$ by the algorithm. 
Then there are two subcases:
\begin{enumerate}
    \item $c\neq i_{j+1}^*$. 
    In this case, the condition at \cref{line:12} must be true, since otherwise the algorithm would store $c$, rather than $i_{j+1}^*$, into  $i_{j+1}$. 
    \item $c=i_{j+1}^*$. In this case, the condition at \cref{line:12} must be false, since otherwise the algorithm would increase $c$ at \cref{line:14} and then either choose $i_{j+1}$ to be greater than $i_{j+1}^*$ or not choose it at all.
\end{enumerate}

Now consider the case $j=j^*$. In this situation, $j$ has reached $j^*$. However, recall that the simulation ends when $c=c^*$. Again, we have two subcases to consider.
\begin{enumerate}
    \item $c\neq c^*$. In this case, the condition at \cref{line:12} must be true, so that the algorithm can continue to increase $c$ to $c^*$ before choosing $i_{j^*+1}$.
    \item $c=c^*$. This is the end of the simulation where we read off $Q_1,\dots,Q_s$ and $a$. Note that the next time the algorithm reaches \cref{line:12}, either $j$ will be greater than $j^*$ or $c$ will be greater than $c^*$.
\end{enumerate}
As we can determine $Q_1,\dots,Q_s$ and $a$ without knowing $\alpha_1,\dots,\alpha_n$, the claim follows.
\end{proof}

\begin{figure}
    \centering
    
        \begin{align*}
       \left[\,\,\,\,
        \begin{matrix}
            \vdots &&&\vdots&\vdots &&\vdots\\
            \rowcolor{blue!10} 1 & X_{i+1} & \cdots & X_{i+1}^{k-1} & X_{i} & \cdots & X_{i}^{k-1} \\
            \rowcolor{blue!20} 1 & X_{i+2} & \cdots & X_{i+2}^{k-1} & X_{i+1} & \cdots & X_{i+1}^{k-1} \\
            \rowcolor{blue!30} 1 & X_{i+3} & \cdots & X_{i+3}^{k-1} & X_{i+2} & \cdots & X_{i+2}^{k-1} \\
            \rowcolor{blue!40} 1 & X_{i+4} & \cdots & X_{i+4}^{k-1} & X_{i+3} & \cdots & X_{i+3}^{k-1} \\
            \vdots &&&\vdots&\vdots &&\vdots\\
        \end{matrix}
        \,\,\,\,
        \right]
        \rightarrow
        \left[\,\,\,\,
        \begin{matrix}
            \vdots &&&\vdots&\vdots &&\vdots\\
            \rowcolor{gray!30} 1 & X_{j+1} & \cdots & X_{j+1}^{k-1} & X_{j} & \cdots & X_{j}^{k-1} \\
            \rowcolor{gray!45} 1 & X_{j+2} & \cdots & X_{j+2}^{k-1} & X_{j+1} & \cdots & X_{j+1}^{k-1} \\
            \rowcolor{gray!60} 1 & X_{j+3} & \cdots & X_{j+3}^{k-1} & X_{j+2} & \cdots & X_{j+2}^{k-1} \\
            \rowcolor{gray!75} 1 & X_{j+4} & \cdots & X_{j+4}^{k-1} & X_{j+3} & \cdots & X_{j+3}^{k-1} \\
            \vdots &&&\vdots&\vdots &&\vdots\\
        \end{matrix}
        \,\,\,\,
        \right]
    \end{align*} 

     \begin{tikzpicture}[scale=0.5]
       \foreach \i in {1,...,30} {
            \node[chainnode] (j\i) at (1*\i,0) {};
            \node[chainnode] (i\i) at (1*\i,1) {};
        }       
        \draw[chain, color=blue!10] (1,0)  -- (2,1);
        \foreach \i in {2,...,4} {
           \pgfmathtruncatemacro{\ii}{10+(\i-1)*10};   
            \draw[chain2, color=blue] (1*\i,0)  -- (1*\i,1);
            \draw[chain, color=blue!\ii] (1*\i,0)  -- (1*\i+1,1);
        }
        \node at (1,-1) {$i$};
        \node at (5,-1) {$i+4$};

        \draw[chain, color=gray!30] (24,0)  -- (25,1);
        \foreach \i in {25,...,28} {
            \pgfmathtruncatemacro{\ii}{30+(\i-24)*15};
            \draw[chain2, color=gray] (1*\i,0)  -- (1*\i,1);
            \draw[chain, color=gray!\ii] (1*\i,0)  -- (1*\i+1,1);
        }
        \node at (24,-1) {$j$};
        \node at (28,-1) {$j+4$};
    \end{tikzpicture}

        \begin{align*}
       \left[\,\,\,\,
        \begin{matrix}
            \vdots &&&\vdots&\vdots &&\vdots\\
            \rowcolor{blue!10} 1 & X_{i+1} & \cdots & X_{i+1}^{k-1} & X_{i} & \cdots & X_{i}^{k-1} \\
            \rowcolor{blue!20} 1 & X_{i+2} & \cdots & X_{i+2}^{k-1} & X_{i+1} & \cdots & X_{i+1}^{k-1} \\
            \rowcolor{blue!30} 1 & X_{i+3} & \cdots & X_{i+3}^{k-1} & X_{i+2} & \cdots & X_{i+2}^{k-1} \\
            \rowcolor{blue!40} 1 & X_{i+4} & \cdots & X_{i+4}^{k-1} & X_{i+3} & \cdots & X_{i+3}^{k-1} \\
            \vdots &&&\vdots&\vdots &&\vdots\\
        \end{matrix}
        \,\,\,\,
        \right]
        \rightarrow
        \left[\,\,\,\,
        \begin{matrix}
            \vdots &&&\vdots&\vdots &&\vdots\\
            \rowcolor{gray!75} 1 & X_{j+3} & \cdots & X_{j+3}^{k-1} & X_{j+4} & \cdots & X_{j+4}^{k-1} \\
            \rowcolor{gray!60} 1 & X_{j+2} & \cdots & X_{j+2}^{k-1} & X_{j+3} & \cdots & X_{j+3}^{k-1} \\
            \rowcolor{gray!45} 1 & X_{j+1} & \cdots & X_{j+1}^{k-1} & X_{j+2} & \cdots & X_{j+2}^{k-1} \\
            \rowcolor{gray!30} 1 & X_{j} & \cdots & X_{j}^{k-1} & X_{j+1} & \cdots & X_{j+1}^{k-1} \\
            \vdots &&&\vdots&\vdots &&\vdots\\
        \end{matrix}
        \,\,\,\,
        \right]
    \end{align*} 
            
         \begin{tikzpicture}[scale=0.5]
       \foreach \i in {1,...,30} {
            \node[chainnode] (j\i) at (1*\i,0) {};
            \node[chainnode] (i\i) at (1*\i,1) {};
        }       
        \draw[chain, color=blue!10] (1,0)  -- (2,1);
        \foreach \i in {2,...,4} {
            \pgfmathtruncatemacro{\ii}{10+(\i-1)*10};
            \draw[chain2, color=blue] (1*\i,0)  -- (1*\i,1);
            \draw[chain, color=blue!\ii] (1*\i,0)  -- (1*\i+1,1);
        }
        \node at (1,-1) {$i$};
        \node at (5,-1) {$i+4$};

        \draw[chain, color=gray!30] (25,0)  -- (24,1);
        \foreach \i in {25,...,28} {
            \pgfmathtruncatemacro{\ii}{30+(\i-24)*15};
            \draw[chain2, color=gray] (1*\i,0)  -- (1*\i,1);
            \draw[chain, color=gray!\ii] (1*\i+1,0)  -- (1*\i,1);
        }
        \node at (24,-1) {$j$};
        \node at (28,-1) {$j+4$};
    \end{tikzpicture}
        \begin{align*}
       \left[\,\,\,\,
        \begin{matrix}
            \vdots &&&\vdots&\vdots &&\vdots\\
            \vdots &&&\vdots&\vdots &&\vdots\\
            \rowcolor{blue!30} 1 & X_{i} & \cdots & X_{i}^{k-1} & X_{i} & \cdots & X_{i}^{k-1} \\
            \vdots &&&\vdots&\vdots &&\vdots\\
            \vdots &&&\vdots&\vdots &&\vdots
        \end{matrix}
        \,\,\,\,
        \right]
        \rightarrow
        \left[\,\,\,\,
        \begin{matrix}
            \vdots &&&\vdots&\vdots &&\vdots\\
            \vdots &&&\vdots&\vdots &&\vdots\\
            \rowcolor{gray!30} 1 & X_{j} & \cdots & X_{j}^{k-1} & X_{j+1} & \cdots & X_{j+1}^{k-1} \\
            \vdots &&&\vdots&\vdots &&\vdots\\
            \vdots &&&\vdots&\vdots &&\vdots
        \end{matrix}
        \,\,\,\,
        \right]
    \end{align*}
    
         \begin{tikzpicture}[scale=0.5]
       \foreach \i in {1,...,30} {
            \node[chainnode] (j\i) at (1*\i,0) {};
            \node[chainnode] (i\i) at (1*\i,1) {};
        }       
        \draw[chain, color=blue!30] (2,0)  -- (2,1);
        \node at (2,-1) {$i$};

        \draw[chain, color=gray!30] (25,0)  -- (24,1);
        \foreach \i in {25,...,28} {
            \pgfmathtruncatemacro{\ii}{30+(\i-24)*15};
            \draw[chain2, color=gray] (1*\i,0)  -- (1*\i,1);
            \draw[chain, color=gray!\ii] (1*\i+1,0)  -- (1*\i,1);
        }
        \node at (24,-1) {$j$};
    \end{tikzpicture}

    \caption{\cref{lem:nonsingular}. Re-indeterminating a faulty (blue) chain with (part of) a (gray) chain from the bank. Three cases: 
    (1) The faulty chain and the bank's chain are both type II with the same orientation.  The new $V$-matrix is equivalent to the original one. 
    (2) The faulty chain and the bank's (gray) chain are both type II with different orientations. The new $V$-matrix is equivalent to the original one when we view the bank's gray-chain in \emph{reverse}.
    (3) The faulty chain is type I and the bank's chain is type II. The new $V$-matrix is not exactly equivalent to the original one, but is similar enough: since the old matrix is full rank (before setting $X_i$), the new matrix certainly is full rank as well, which is what we need.
    }
    \label{fig:replace}
\end{figure}

We also need the following crucial lemma.

\begin{lemma}\label{lem:nonsingular}
At \cref{line:12} of \cref{alg:algorithm-linear}, $M|_{X_i=\alpha_i \text{ for } i\in S}$ is nonsingular.
\end{lemma}

\begin{proof}
We prove by induction on the iteration of the while loop. When the algorithm first reaches \cref{line:12}, we have $S=\emptyset$, and hence $M|_{X_i=\alpha_i \text{ for } i\in S}$ is simply $M$, which is the top $(2k-1)\times (2k-1)$ submatrix of $\widetilde{V}_{k,I,J,(Q_i)_{i=1}^{s-m}}=\widetilde{V}_{k,I,J,(P_i)_{i=1}^{s-m}}$.
Note that by \cref{defn:blocks}, this matrix equals, up to reordering of rows, the $V$-matrix $V_{k,2k-1,I^P,J^P}$ for some $P\subseteq \bigcup_{i=1}^{s-m} P_i$ of size $2k-1$.
So by \cref{roub}, $M|_{X_i=\alpha_i \text{ for } i\in S}$ is nonsingular.

Next, assume that for some $t\in\N^+$, $M|_{X_i=\alpha_i \text{ for } i\in S}$ is nonsingular at \cref{line:12} in the $t$-th iteration of the while loop. It suffices to show that $M|_{X_i=\alpha_i \text{ for } i\in S}$ is still nonsingular at \cref{line:12} in the $(t+1)$-th iteration. There are two cases:

Case 1: the matrix $\overline{M}=M|_{X_i=\alpha_i \text{ for } i\in S'}$ is nonsingular at \cref{line:12} in the $t$-th iteration.
The algorithm then runs \crefrange{line:13}{line:15} and replaces $S$ by $S'$.
So in the $(t+1)$-th iteration, $M|_{X_i=\alpha_i \text{ for } i\in S}$ is nonsingular at \cref{line:12}.

Case 2: the matrix $\overline{M}=M|_{X_i=\alpha_i \text{ for } i\in S'}$ is singular at \cref{line:12} in the $t$-th iteration.
In this case, the algorithm runs \crefrange{line:17}{line:24}, which do not change the set $S$. However, $Q_a$ is replaced by a subset of $Q_b$ at \cref{line:17}, which we denote by $Q_b'$. This replacement, together with \cref{line:9} in the $(t+1)$-th iteration, changes the value of $M$. We need to verify that it does not affect the nonsingularity of $M|_{X_i=\alpha_i \text{ for } i\in S}$. 

At \cref{line:12} in the $t$-th iteration, we have $\var(I,J,Q_a)\cap S=\emptyset$. This is because if some $x\in\var(I,J, Q_a)$ is in $S$, then it must have been added to $S'$ before the $t$-th iteration, contradicting \cref{lem:new-elements}.
Similarly, we have  $\var(I,J,Q_b')\cap S=\emptyset$ by the same argument and the fact that $Q_b'$ becomes the new $Q_a$ at \cref{line:9} in the $(t+1)$-th iteration. 
The conclusion is that even under the assignment $X_i=\alpha_i$ for $i\in S$, the variables $X_{i'}$ with $i'\in \var(I,J,Q_a)\cup \var(I,J,Q_b')$ remain free variables.

Let $e=|Q_a|$.
Suppose $Q_a=\{j_1,\dots,j_e\}\subseteq [\ell]$ with $j_1<\dots<j_e$, which is replaced by $Q_b'=\{j_1',\dots,j_e'\}\subseteq Q_b$ with $j_1'<\dots<j_e'$ at \cref{line:17} in the $t$-th iteration.
As $M$ only contains the first $2k-1$ rows of $\widetilde{V}_{k,I,J,(Q_i)_{i=1}^{s-m}}$,
it only contains the first $e'$ rows of $V_{k,e,I^{Q_a},J^{Q_a}}$ for some $e'\in\{0,\dots,e\}$. If $e'=0$, then replacing $Q_a$ by $Q_b'$ does not change $M$ and we are done. So assume $e'>0$.

Define $M'$ as the matrix $M|_{X_i=\alpha_i \text{ for } i\in S}$
at \cref{line:12} in the $t$-th iteration,
and define $M''$ as $M|_{X_i=\alpha_i \text{ for } i\in S}$ at \cref{line:12} in the $(t+1)$-th iteration. By assumption, $M'$ is nonsingular, and we want to show that $M''$ is also nonsingular.
Define the increasing subsequences $I',J',I'',J''\in [n]^{e'}$
by $I'_i=I_{j_i}$, $J'_i=J_{j_i}$, $I''_i=I_{j_i'}$, and $J''_i=J_{j_i'}$ for $i\in [e']$, i.e., $I'=I^{Q_a}$, $J'=J^{Q_a}$, $I''=I^{Q_b'}$, and $J''=J^{Q_b'}$.
Then replacing $Q_a$ by $Q_b'$ has the effect that the submatrix $V_{k,e',I',J'}$ of $M'$ is replaced by 
the submatrix $V_{k,e',I'',J''}$ of
$M''$, turning $M'$ into $M''$.

Note that $(I',J')$ is a chain of the same type as $(I^{Q_a}, I^{Q_a})$, and $(I'',J'')$ is a chain of the same type as $(I^{Q_b},J^{Q_b})$.
There are several cases (see Figure~\ref{fig:replace}):
\begin{enumerate}
\item $(I', J')$ and $(I'', J'')$ are both Type-II chains, and $I'_1<J'_1$ iff $I''_1<J_1''$. 
In this case, if we substitute $X_{I'_i}$ for $X_{I''_i}$ and $X_{J'_i}$ for $X_{J''_i}$ in $M''$ for $i\in [e']$, then we recover $M'$. That is, $M'=M''|_{X_{I''_i}=X_{J'_i}, X_{J''_i}=X_{J'_i} \text{ for } i\in [e']}$. Note that either $I''_i=J''_{i+1}$ for $i=1,\dots,e'-1$ or $I''_{i+1}=J''_i$ for $i=1,\dots,e'-1$, but the same relations hold for the coordinates of $I'$ and $J'$ as well since $(I',J')$ and $(I'',J'')$ have the same type and the same orientation. So the substitutions of variables are valid.

\item $(I', J')$ and $(I'', J'')$ are both Type-II chains, and $I'_1<J'_1$ iff $I''_1>J_1''$. 
In this case, if we substitute $X_{I'_{e-i+1}}$ for $X_{I''_i}$ and $X_{J'_{e-i+1}}$ for $X_{J''_i}$ in $M''$ for $i\in [e']$, then we recover $M'$ up to reordering of rows. That is, $M'$ is, up to reordering of rows, equal to the matrix $M''|_{X_{I''_i}=X_{J'_{e'-i+1}}, X_{J''_i}=X_{J'_{e'-i+1}} \text{ for } i\in [e']}$.
Note that if $I''_i=J''_{i+1}$, then $I'_{e-i+1}=J'_{e-i}$. Similarly, if $I''_{i+1}=J''_{i}$, then $I'_{e-i}=J'_{e-i+1}$. That is, the relations satisfied by the coordinates of $I''$ and $J''$ are also satisfied by the coordinates of the \emph{reversals} of $I'$ and $J'$. So the substitutions of variables are valid.

\item $(I', J')$ is a Type-I chain. In this case, $e'=1$ and $I_1'=J_1'$.
Then $M'=M''|_{X_{I_1''}=X_{I_1'}, X_{J_1''}=X_{J_1'}}$. Note that $(I'',J'')$ may or may not be a Type-I chain, but the substitutions of variables are valid regardless.
\end{enumerate}

Finally, if $(I'', J'')$ is a Type-I chain, then by \cref{order-on-Pi}, the chain $(I',J')$, which comes from $(I^{P_{b'}}, J^{P_{b'}})$ for some $b'<b$, is also a Type-I chain. So the above three cases cover all the possibilities.
In each case, $M'$ is, up to reordering rows, equal to $\sigma(M'')$, where $\sigma$ is a ring endomorphism of $\Fq[X_i: i\in [n]\setminus S]$ induced by the substitutions of variables described above and applied to $M''$ entrywise.
So $\det(M')=\pm \sigma(\det(M''))$.
As $\det(M')\neq 0$, we have $\det(M'')\neq 0$, i.e., $M''$ is nonsingular. 
\end{proof}

Now we are ready to bound the probability that \cref{alg:algorithm-linear} outputs a particular sequence.

\begin{lemma}\label{lastlem}
Suppose $q\geq n$ and $(\alpha_1,\dots,\alpha_n)$ is chosen uniformly at random from the set of all $n$-tuples of distinct elements in $\F_q$. 
The probability that \cref{alg:algorithm-linear} outputs a fixed sequence $(i_1^*,\dots,i_r^*)\in \{0,\dots,2k-2\}^r$ is at most $p^r$, where
\[
p=\frac{((1/\eps_0)+1)\cdot 2(k-1)}{q-n+1}.
\]
\end{lemma}

\begin{proof}
For $0\leq t\leq r$, let $F_{t}$ be the event that the algorithm chooses $(i_1,\dots,i_t)=(i_1^*,\dots,i_{t}^*)$. 
We will prove by induction that $\Pr[F_t]\leq p^{t}$. The lemma follows by choosing $t=r$.

For $t=0$, the claim hold trivially. Assume the claim holds for some $t<r$. We now prove that it holds for $t+1$ as well.

Let $F'_t$ be the sub-event of $F_t$ that the algorithm reaches \cref{line:12} with $j=t$, $(i_1,\dots,i_t)=(i^*_1,\dots,i^*_t)$, and $c=i^*_{t+1}$.
Note that $F'_t$ must happen if $F_{t+1}$ happens, i.e., $F_{t+1}$ is a sub-event of $F'_t$. 
If $\Pr[F'_t]=0$, then $\Pr[F_{t+1}]=0$ and we are done. So assume this is not the case.
By the induction hypothesis, we have $\Pr[F'_t]\leq \Pr[F_t]\leq p^{t}$.
So it suffices to bound the conditional probability $\Pr[F_{t+1}|F_t']$ by $p$.

Condition on the event $F_t'$. Consider the moment when algorithm reaches \cref{line:12} with $j=t$, $(i_1,\dots,i_t)=(i^*_1,\dots,i^*_t)$, and $c=i^*_{t+1}$.
By \cref{lem:determined}, $Q_1,\dots,Q_s$ and $a$ are determined.
The set $S$ of assigned variables
is also determined via $S=\bigcup_{i=1}^{a-1} \var(I,J, Q_i)$.

We further fix arbitrary values of $\alpha_i$ for $i\in [n]\setminus \var(I,J,Q_a)$ consistent with $F_t'$.
It suffices to show that the probability that $F_{t+1}$ occurs (conditioned on $F'_t$ and the fixed $\alpha_i$ for $i\in [n]\setminus \var(I,J,Q_a)$) is at most $p$.

By \cref{lem:new-elements}, the elements $\alpha_i$ for $i\in \var(I,J,Q_a)$ have not been used in the previous iterations of the while loop.
Thus, the (conditional) distribution of $(\alpha_i)_{i\in\var(I,J,Q_a)}$ is the uniform distribution over the set of $|\var(I,J,Q_a)|$-tuples with distinct coordinates in the set $T:=\Fq\setminus \{\alpha_i: i\in [n]\setminus \var(I,J,Q_a)\}$. 

Let $M^*:=M|_{X_i=\alpha_i \text{ for } i\in S}$.
Then $\overline{M}=M^*|_{X_i=\alpha_i \text{ for } i\in \var(I,J,Q_a)}$.
By \cref{lem:nonsingular}, $M^*$ is nonsingular, i.e., $\det(M^*)\neq 0$.
View $\det(M^*)$ as a polynomial in the variables $X_i$, $i\in [n]\setminus (S\cup \var(I,J,Q_a))$, over the polynomial ring $\Fq[X_i: i\in \var(I,J,Q_a)]$. Let $Q\in \Fq[X_i: i\in \var(I,J,Q_a)]$ be the coefficient of a nonzero term of $\det(M^*)$. Then $Q\neq 0$. We also have 
\[
\deg_{X_i}(Q)\leq \deg_{X_i}(\det(M^*))\leq 2(k-1)
\]
for $i\in \var(I,J,Q_a)$, where the last inequality holds since $X_i$ appears in at most two rows of $M^*$ and every entry of $M^*$ containing $X_i$ has the form $X_i^d$ with $d\leq k-1$. By \cref{lem:zero-prob}, we have 
\[
\Pr[Q(\alpha_i: i\in \var(I,J,Q_a))=0]\leq \frac{|\var(I,J,Q_a)|\cdot 2(k-1)}{|T|-|\var(I,J,Q_a)|+1}
\leq \frac{((1/\eps_0)+1)\cdot 2(k-1)}{q-n+1}=p,
\]
where the second inequality uses the facts that $|\var(I,J,Q_a)|\leq |Q_a|+1\leq (1/\eps_0)+1$ (as $(I^{Q_a},J^{Q_a})$ is a chain) and that the size of $T=\Fq\setminus \{\alpha_i: i\in [n]\setminus \var(I,J,Q_a)\}$ is $q-n+|\var(I,J,Q_a)|$.

Note that $Q(\alpha_i: i\in \var(I,J,Q_a))$ is the coefficient of a term of $\det(M^*)|_{X_i=\alpha_i \text{ for } i\in\var(I,J,Q_a)}=\det(\overline{M})$.
So if $\det(\overline{M})$ is zero, so is $Q(\alpha_i: i\in \var(I,J,Q_a))$.
This implies that the probability that $\overline{M}$ is singular is at most $p$.
Note that $F_{t+1}$ occurs only if $\overline{M}$ is singular.
(This is because, assuming $\overline{M}$ is not singular, the algorithm will not run \crefrange{line:17}{line:24} and store $c=i_{t+1}^*$ into $i_{t+1}$. Instead, it will run \crefrange{line:13}{line:15} and increase $c$.) Therefore, the (conditional) probability that $F_{t+1}$ occurs is at most $p$, as desired.
\end{proof}

Taking the union bound over all possible output sequences, we obtain the following corollary.

\begin{cor}\label{cor:prob-bound} 
Suppose $q\geq n$ and $(\alpha_1,\dots,\alpha_n)$ is chosen uniformly at random from the set of all $n$-tuples of distinct elements in $\F_q$.
Also suppose $\ell, r\in\N^+$ and $\eps_0\in (0,1)$ satisfy \eqref{eq:condition}.
Then the probability that $V_{k,\ell,I,J}|_{X_1=\alpha_1,\dots,X_n=\alpha_n}$ does not have full column rank is at most $2^{2k+r-2}p^r$, where $p=\frac{((1/\eps_0)+1)\cdot 2(k-1)}{q-n+1}$.
\end{cor}

\begin{proof}
By \cref{lem:output}, any sequence $(i_1,\dots,i_r)\in\{0,\dots,2k-2\}^r$ output by \cref{alg:algorithm-linear} satisfies $i_1\leq \dots\leq i_r$.
The number of such sequences is $\binom{(2k-1)+r-1}{r}\leq 2^{2k+r-2}$.
By \cref{lastlem}, any fixed sequence is output with probability at most $p^r$.
So by the union bound,
the probability that \cref{alg:algorithm-linear} outputs a sequence $(i_1,\dots,i_r)$ is at most $2^{2k+r-2} p^r$. 

By \cref{lem:terminate}, 
\cref{alg:algorithm-linear} outputs a sequence $(i_1,\dots,i_r)$ whenever $V_{k,\ell,I,J}|_{X_1=\alpha_1,\dots,X_n=\alpha_n}$ does not have full column rank. 
So the probability of the latter event is at most $2^{2k+r-2}\cdot p^r$. 
\end{proof}
Finally, we prove our main theorem below.
\begin{thm}[Detailed version of \cref{main2}]\label{main}
Let $\varepsilon\in (0,1)$ and $n,k\in\N^+$, where $k\leq n$.  
Let $q$ be a prime power such that $q\geq n+2^{c/\eps^2} k$, where $c>0$ is a large enough absolute constant.
Suppose $(\alpha_1,\dots,\alpha_n)$ is chosen uniformly at random from the set of all $n$-tuples of distinct elements in $\F_q$.
Then with probability at least $1-2^{-n}>0$, the code $\text{RS}_{n,k}(\alpha_1, \ldots, \alpha_n)$ over $\Fq$ corrects at least $(1-\epsilon)n-2k+1$ adversarial insdel errors.
\end{thm}
\begin{proof}
Let $c_0>0$ be a large enough absolute constant, and let $c=c_0^2$.
First assume $\eps^2 n\leq c_0$. Since $c/\eps^2\geq c_0 n$, it suffices to prove the claim for $q\geq n+2^{c_0 n}k=2^{\Theta(n)}$.
In this case, the code can, in fact, correct at least $n-2k+1$ adversarial insdel errors. A proof was implicitly given in \cite[proof of Theorem 16]{con2023reed} and also sketched in our \cref{sec:proof-overview}.

Next, assume $\eps^2 n\geq c_0$.
Choose $\ell=(2k-1) + \floor{\varepsilon n}$, $\eps_0=\eps/4$, and $r=\lceil\frac{\eps^2 n}{16}\rceil$. It is easy to verify that that \eqref{eq:condition} holds given that the constant $c_0$ is large enough.

By Lemma \ref{bad}, if $\mathsf{RS}_{n,k}(\alpha_1,\dots,\alpha_n)\subseteq\F_q^n$ fails to correct $n-\ell$ adversarial insdel errors, then we will have two increasing subsequences $I,J\in[n]^{\ell}$ which agree on at most $k-1$ coordinates, such that the matrix $V_{k,\ell,I, J}|_{X_{1}=\alpha_{1},\dots,X_{n}=\alpha_{n}}$ does not have full column rank. On the other hand, applying \cref{cor:prob-bound} with the chosen parameters $\ell$, $\eps_0$, and $r$, 
for fixed $I,J$, we have that
\begin{align*}
\Pr\left[\begin{aligned}
    &\text{The matrix }V_{k,\ell,I,J}|_{X_1=\alpha_1,\dots,X_{n}=\alpha_{n}}\\
&\text{does not have full column rank}
\end{aligned}\right]&\leq 
2^{2k+\lceil \frac{\varepsilon^2 n}{16}\rceil-2}\cdot\left(\frac{((4/\eps)+1)\cdot 2(k-1)}{q-n+1}\right)^{\lceil\frac{\varepsilon^2 n}{16}\rceil}\\
&\leq 2^{3n}\cdot \left(\frac{((4/\eps)+1)\cdot 2(k-1)}{q-n+1}\right)^{\lceil\frac{\varepsilon^2 n}
{16}\rceil}\\
&\leq 2^{-3n},
\end{align*}
where the last inequality uses the facts that $q\geq n+2^{c/\eps^2} k$ and that $c>0$ is a large enough constant.
By the union bound over all $(I,J)$,
the probability that $\mathsf{RS}_{n,k}(\alpha_1,\dots,\alpha_n)$ cannot correct $(1-\epsilon)n-2k+1$ adversarial insdel errors is at most $\binom{n}{\ell}^2\cdot 2^{-3n}\leq 2^{-n}$, as desired.
\end{proof}

\bibliographystyle{alpha}
\bibliography{ref}

\newcommand{\etalchar}[1]{$^{#1}$}
\begin{thebibliography}{SWZGY17}

\bibitem[AGFC07]{abdel2007linear}
Khaled~AS Abdel-Ghaffar, Hendrik~C Ferreira, and Ling Cheng.
\newblock On linear and cyclic codes for correcting deletions.
\newblock In {\em 2007 IEEE International Symposium on Information Theory
  (ISIT)}, pages 851--855. IEEE, 2007.

\bibitem[AGL24]{alrabiah2023randomly}
Omar Alrabiah, Venkatesan Guruswami, and Ray Li.
\newblock Randomly punctured {R}eed--{S}olomon codes achieve list-decoding
  capacity over linear-sized fields.
\newblock In {\em Proceedings of the 56th Annual ACM Symposium on Theory of
  Computing (STOC)}, pages 1458--1469, 2024.

\bibitem[BDGZ24]{brakensiek2024ag}
Joshua Brakensiek, Manik Dhar, Sivakanth Gopi, and Zihan Zhang.
\newblock {AG} codes achieve list decoding capacity over constant-sized fields.
\newblock In {\em Proceedings of the 56th Annual ACM Symposium on Theory of
  Computing (STOC)}, pages 740--751, 2024.

\bibitem[BGH16]{bukh2016improved}
Boris Bukh, Venkatesan Guruswami, and Johan H{\aa}stad.
\newblock An improved bound on the fraction of correctable deletions.
\newblock {\em IEEE Transactions on Information Theory}, 63(1):93--103, 2016.

\bibitem[BGM23]{brakensiek2023generic}
Joshua Brakensiek, Sivakanth Gopi, and Visu Makam.
\newblock Generic {R}eed-{S}olomon codes achieve list-decoding capacity.
\newblock In {\em Proceedings of the 55th Annual ACM Symposium on Theory of
  Computing (STOC)}, pages 1488--1501, 2023.

\bibitem[BGZ17]{brakensiek2017efficient}
Joshua Brakensiek, Venkatesan Guruswami, and Samuel Zbarsky.
\newblock Efficient low-redundancy codes for correcting multiple deletions.
\newblock {\em IEEE Transactions on Information Theory}, 64(5):3403--3410,
  2017.

\bibitem[CGHL21]{cheng2020efficient}
Kuan Cheng, Venkatesan Guruswami, Bernhard Haeupler, and Xin Li.
\newblock Efficient linear and affine codes for correcting
  insertions/deletions.
\newblock In D{\'{a}}niel Marx, editor, {\em Proceedings of the 2021 {ACM-SIAM}
  Symposium on Discrete Algorithms (SODA)}, pages 1--20. {SIAM}, 2021.

\bibitem[Che22]{chen2022coordinate}
Hao Chen.
\newblock Coordinate-ordering-free upper bounds for linear insertion-deletion
  codes.
\newblock {\em IEEE Transactions on Information Theory}, 68(8):5126--5132,
  2022.

\bibitem[CJL{\etalchar{+}}23]{cheng2023linear}
Kuan Cheng, Zhengzhong Jin, Xin Li, Zhide Wei, and Yu~Zheng.
\newblock Linear insertion deletion codes in the high-noise and high-rate
  regimes.
\newblock In {\em 50th International Colloquium on Automata, Languages, and
  Programming (ICALP)}. Schloss-Dagstuhl-Leibniz Zentrum f{\"u}r Informatik,
  2023.

\bibitem[CJLW22]{cheng2018deterministic}
Kuan Cheng, Zhengzhong Jin, Xin Li, and Ke~Wu.
\newblock Deterministic document exchange protocols and almost optimal binary
  codes for edit errors.
\newblock {\em Journal of the ACM}, 69(6):1--39, 2022.

\bibitem[CR03]{crochemore2003jewels}
Maxime Crochemore and Wojciech Rytter.
\newblock {\em Jewels of stringology: text algorithms}.
\newblock World Scientific, 2003.

\bibitem[CR20]{cheraghchi2020overview}
Mahdi Cheraghchi and Jo{\~a}o Ribeiro.
\newblock An overview of capacity results for synchronization channels.
\newblock {\em IEEE Transactions on Information Theory}, 67(6):3207--3232,
  2020.

\bibitem[CST22]{con2022explicit}
Roni Con, Amir Shpilka, and Itzhak Tamo.
\newblock Explicit and efficient constructions of linear codes against
  adversarial insertions and deletions.
\newblock {\em IEEE Transactions on Information Theory}, 68(10):6516--6526,
  2022.

\bibitem[CST23]{con2023reed}
Roni Con, Amir Shpilka, and Itzhak Tamo.
\newblock {R}eed--{S}olomon codes against adversarial insertions and deletions.
\newblock {\em IEEE Transactions on Information Theory}, 2023.

\bibitem[CST24]{con2023optimal}
Roni Con, Amir Shpilka, and Itzhak Tamo.
\newblock Optimal two-dimensional {R}eed--{S}olomon codes correcting insertions
  and deletions.
\newblock {\em IEEE Transactions on Information Theory}, 2024.

\bibitem[CZ22]{chen2021improved}
Bocong Chen and Guanghui Zhang.
\newblock Improved singleton bound on insertion-deletion codes and optimal
  constructions.
\newblock {\em IEEE Transactions on Information Theory}, 68(5):3028--3033,
  2022.

\bibitem[DLTX21]{duc2019explicit}
Tai~Do Duc, Shu Liu, Ivan Tjuawinata, and Chaoping Xing.
\newblock Explicit constructions of two-dimensional {R}eed-{S}olomon codes in
  high insertion and deletion noise regime.
\newblock {\em IEEE Transactions on Information Theory}, 67(5):2808--2820,
  2021.

\bibitem[Eli57]{elias1957list}
Peter Elias.
\newblock List decoding for noisy channels.
\newblock {\em Wescon Convention Record, Part 2, Institute of Radio Engineers},
  pages 99--104, 1957.

\bibitem[FKS22]{ferber2022list}
Asaf Ferber, Matthew Kwan, and Lisa Sauermann.
\newblock List-decodability with large radius for {R}eed-{S}olomon codes.
\newblock In {\em 2021 IEEE 62nd Annual Symposium on Foundations of Computer
  Science (FOCS)}, pages 720--726. IEEE, 2022.

\bibitem[GBC{\etalchar{+}}13]{goldman2013towards}
Nick Goldman, Paul Bertone, Siyuan Chen, Christophe Dessimoz, Emily~M LeProust,
  Botond Sipos, and Ewan Birney.
\newblock Towards practical, high-capacity, low-maintenance information storage
  in synthesized {DNA}.
\newblock {\em Nature}, 494(7435):77--80, 2013.

\bibitem[GH21]{guruswami2021explicit}
Venkatesan Guruswami and Johan H{\aa}stad.
\newblock Explicit two-deletion codes with redundancy matching the existential
  bound.
\newblock In {\em Proceedings of the 2021 ACM-SIAM Symposium on Discrete
  Algorithms (SODA)}, pages 21--32. SIAM, 2021.

\bibitem[GHL22]{guruswami2022zero}
Venkatesan Guruswami, Xiaoyu He, and Ray Li.
\newblock The zero-rate threshold for adversarial bit-deletions is less than
  1/2.
\newblock {\em IEEE Transactions on Information Theory}, 69(4):2218--2239,
  2022.

\bibitem[GL16]{guruswami2016efficiently}
Venkatesan Guruswami and Ray Li.
\newblock Efficiently decodable insertion/deletion codes for high-noise and
  high-rate regimes.
\newblock In {\em 2016 IEEE International Symposium on Information Theory
  (ISIT)}, pages 620--624. IEEE, 2016.

\bibitem[GLS{\etalchar{+}}22]{guo2022improved}
Zeyu Guo, Ray Li, Chong Shangguan, Itzhak Tamo, and Mary Wootters.
\newblock Improved list-decodability and list-recoverability of
  {R}eed-{S}olomon codes via tree packings.
\newblock In {\em 2021 IEEE 62nd Annual Symposium on Foundations of Computer
  Science (FOCS)}, pages 708--719. IEEE, 2022.

\bibitem[GS98]{guruswami1998improved}
Venkatesan Guruswami and Madhu Sudan.
\newblock Improved decoding of {R}eed-{S}olomon and algebraic-geometric codes.
\newblock In {\em Proceedings 39th Annual Symposium on Foundations of Computer
  Science (FOCS)}, pages 28--37. IEEE, 1998.

\bibitem[GST22]{goldberg2022list}
Eitan Goldberg, Chong Shangguan, and Itzhak Tamo.
\newblock List-decoding and list-recovery of {R}eed--{S}olomon codes beyond the
  {J}ohnson radius for every rate.
\newblock {\em IEEE Transactions on Information Theory}, 69(4):2261--2268,
  2022.

\bibitem[GW17]{guruswami2017deletion}
Venkatesan Guruswami and Carol Wang.
\newblock Deletion codes in the high-noise and high-rate regimes.
\newblock {\em IEEE Transactions on Information Theory}, 63(4):1961--1970,
  2017.

\bibitem[GXYZ24]{guo2024random}
Zeyu Guo, Chaoping Xing, Chen Yuan, and Zihan Zhang.
\newblock Random {G}abidulin codes achieve list decoding capacity in the rank
  metric.
\newblock {\em arXiv preprint arXiv:2404.13230}, 2024.
\newblock To appear in FOCS 2024.

\bibitem[GZ23]{GZ23}
Zeyu Guo and Zihan Zhang.
\newblock Randomly punctured {R}eed-{S}olomon codes achieve the list decoding
  capacity over polynomial-size alphabets.
\newblock In {\em 2023 IEEE 64th Annual Symposium on Foundations of Computer
  Science (FOCS)}, pages 164--176. IEEE, 2023.

\bibitem[Hae19]{haeupler2019optimal}
Bernhard Haeupler.
\newblock Optimal document exchange and new codes for insertions and deletions.
\newblock In {\em 2019 IEEE 60th Annual Symposium on Foundations of Computer
  Science (FOCS)}, pages 334--347. IEEE, 2019.

\bibitem[Ham50]{hamming1950error}
Richard~W. Hamming.
\newblock Error detecting and error correcting codes.
\newblock {\em Bell System technical journal}, 29(2):147--160, 1950.

\bibitem[HMG19]{heckel2019characterization}
Reinhard Heckel, Gediminas Mikutis, and Robert~N Grass.
\newblock {A characterization of the DNA data storage channel}.
\newblock {\em Scientific reports}, 9(1):1--12, 2019.

\bibitem[HS21a]{haeupler2021synchronization}
Bernhard Haeupler and Amirbehshad Shahrasbi.
\newblock Synchronization strings and codes for insertions and deletions--a
  survey.
\newblock {\em IEEE Transactions on Information Theory}, 67(6):3190--3206,
  2021.

\bibitem[HS21b]{haeupler2021synchronization-org}
Bernhard Haeupler and Amirbehshad Shahrasbi.
\newblock Synchronization strings: Codes for insertions and deletions
  approaching the {S}ingleton bound.
\newblock {\em Journal of the ACM}, 68(5):1--39, 2021.

\bibitem[HSRD17]{heckel2017fundamental}
Reinhard Heckel, Ilan Shomorony, Kannan Ramchandran, and NC~David.
\newblock Fundamental limits of {DNA} storage systems.
\newblock In {\em 2017 IEEE International Symposium on Information Theory
  (ISIT)}, pages 3130--3134. IEEE, 2017.

\bibitem[Joh62]{johnson1962new}
Selmer Johnson.
\newblock A new upper bound for error-correcting codes.
\newblock {\em IRE Transactions on Information Theory}, 8(3):203--207, 1962.

\bibitem[JZCW23]{ji2023strict}
Qinqin Ji, Dabin Zheng, Hao Chen, and Xiaoqiang Wang.
\newblock Strict half-{S}ingleton bound, strict direct upper bound for linear
  insertion-deletion codes and optimal codes.
\newblock {\em IEEE Transactions on Information Theory}, 2023.

\bibitem[Lev66]{levenshtein1966binary}
Vladimir~I Levenshtein.
\newblock Binary codes capable of correcting deletions, insertions, and
  reversals.
\newblock In {\em Soviet physics doklady}, volume~10, pages 707--710, 1966.

\bibitem[Lev02]{levenshtein2002bounds}
Vladimir~I Levenshtein.
\newblock Bounds for deletion/insertion correcting codes.
\newblock In {\em Proceedings IEEE International Symposium on Information
  Theory (ISIT)}, page 370. IEEE, 2002.

\bibitem[Liu24]{liu2024optimal}
Jingge Liu.
\newblock Optimal {RS} codes and {GRS} codes against adversarial insertions and
  deletions and optimal constructions.
\newblock {\em IEEE Transactions on Information Theory}, 2024.

\bibitem[LSWZY19]{lenz2019coding}
Andreas Lenz, Paul~H Siegel, Antonia Wachter-Zeh, and Eitan Yaakobi.
\newblock Coding over sets for {DNA} storage.
\newblock {\em IEEE Transactions on Information Theory}, 66(4):2331--2351,
  2019.

\bibitem[LT21]{liu20212}
Shu Liu and Ivan Tjuawinata.
\newblock {On 2-dimensional insertion-deletion Reed-Solomon codes with optimal
  asymptotic error-correcting capability}.
\newblock {\em Finite Fields and Their Applications}, 73:101841, 2021.

\bibitem[MBT10]{mercier2010survey}
Hugues Mercier, Vijay~K Bhargava, and Vahid Tarokh.
\newblock A survey of error-correcting codes for channels with symbol
  synchronization errors.
\newblock {\em IEEE Communications Surveys \& Tutorials}, 12(1):87--96, 2010.

\bibitem[Mit09]{mitzenmacher2009survey}
Michael Mitzenmacher.
\newblock A survey of results for deletion channels and related synchronization
  channels.
\newblock {\em Probability Surveys}, 6:1--33, 2009.

\bibitem[RW14]{rudra2014every}
Atri Rudra and Mary Wootters.
\newblock Every list-decodable code for high noise has abundant near-optimal
  rate puncturings.
\newblock In {\em Proceedings of the 46th Annual ACM Symposium on Theory of
  Computing (STOC)}, pages 764--773, 2014.

\bibitem[RZVW24]{ron2024efficient}
Noga Ron-Zewi, S~Venkitesh, and Mary Wootters.
\newblock Efficient list-decoding of polynomial ideal codes with optimal list
  size.
\newblock {\em arXiv preprint arXiv:2401.14517}, 2024.

\bibitem[SB20]{sima2020optimal}
Jin Sima and Jehoshua Bruck.
\newblock On optimal k-deletion correcting codes.
\newblock {\em IEEE Transactions on Information Theory}, 67(6):3360--3375,
  2020.

\bibitem[SGB20]{sima2020systematic}
Jin Sima, Ryan Gabrys, and Jehoshua Bruck.
\newblock Optimal systematic t-deletion correcting codes.
\newblock In {\em 2020 IEEE International Symposium on Information Theory
  (ISIT)}, pages 769--774. IEEE, 2020.

\bibitem[SH22]{shomorony2022information}
Ilan Shomorony and Reinhard Heckel.
\newblock Information-theoretic foundations of {DNA} data storage.
\newblock {\em Foundations and Trends{\textregistered} in Communications and
  Information Theory}, 19(1):1--106, 2022.

\bibitem[Sha48]{shannon1948mathematical}
Claude~Elwood Shannon.
\newblock A mathematical theory of communication.
\newblock {\em Bell system technical journal}, 27(3):379--423, 1948.

\bibitem[SNW02]{safavi2002traitor}
Reihaneh Safavi-Naini and Yejing Wang.
\newblock Traitor tracing for shortened and corrupted fingerprints.
\newblock In {\em ACM workshop on Digital Rights Management}, pages 81--100.
  Springer, 2002.

\bibitem[ST20]{shangguan2020combinatorial}
Chong Shangguan and Itzhak Tamo.
\newblock Combinatorial list-decoding of {R}eed-{S}olomon codes beyond the
  {J}ohnson radius.
\newblock In {\em Proceedings of the 52nd Annual ACM SIGACT Symposium on Theory
  of Computing (STOC)}, pages 538--551, 2020.

\bibitem[Sud97]{sudan1997decoding}
Madhu Sudan.
\newblock Decoding of {R}eed {S}olomon codes beyond the error-correction bound.
\newblock {\em Journal of Complexity}, 13(1):180--193, 1997.

\bibitem[SWZGY17]{schoeny2017codes}
Clayton Schoeny, Antonia Wachter-Zeh, Ryan Gabrys, and Eitan Yaakobi.
\newblock Codes correcting a burst of deletions or insertions.
\newblock {\em IEEE Transactions on Information Theory}, 63(4):1971--1985,
  2017.

\bibitem[SZ99]{schulman1999asymptotically}
Leonard~J Schulman and David Zuckerman.
\newblock Asymptotically good codes correcting insertions, deletions, and
  transpositions.
\newblock {\em IEEE Transactions on Information Theory}, 45(7):2552--2557,
  1999.

\bibitem[TSN07]{tonien2007construction}
Dongvu Tonien and Reihaneh Safavi-Naini.
\newblock {Construction of deletion correcting codes using generalized
  Reed--Solomon codes and their subcodes}.
\newblock {\em Designs, Codes and Cryptography}, 42(2):227--237, 2007.

\bibitem[VT65]{varshamov1965codes}
RR~Varshamov and GM~Tenengolts.
\newblock Codes which correct single asymmetric errors (in {R}ussian).
\newblock {\em Automatika i Telemkhanika}, 161(3):288--292, 1965.

\bibitem[WMSN04]{wang2004deletion}
Yejing Wang, Luke McAven, and Reihaneh Safavi-Naini.
\newblock {Deletion correcting using generalized Reed-Solomon codes}.
\newblock In {\em Coding, Cryptography and Combinatorics}, pages 345--358.
  Springer, 2004.

\bibitem[Woz58]{wozencraft1958list}
John~M Wozencraft.
\newblock List decoding.
\newblock {\em Quarterly Progress Report}, 48:90--95, 1958.

\bibitem[Yas24]{yasunaga2024improved}
Kenji Yasunaga.
\newblock Improved bounds for codes correcting insertions and deletions.
\newblock {\em Designs, Codes and Cryptography}, pages 1--12, 2024.

\end{thebibliography}

\end{document}